\documentclass[journal]{IEEEtran}
\usepackage{amsmath,amsfonts}
\usepackage{algorithmic}
\usepackage{algorithm}
\usepackage{array}
\usepackage[caption=false,font=normalsize,labelfont=sf,textfont=sf]{subfig}
\usepackage{textcomp}
\usepackage{stfloats}
\usepackage{url}
\usepackage{verbatim}
\usepackage{graphicx}
\usepackage{cite}
\usepackage{xspace}
\usepackage{xcolor}
\usepackage[hidelinks]{hyperref}
\newcommand{\T}{\hspace{1.5pt}{\mathcal{T}}}
\newcommand{\wi}{w_{I}}
\newcommand{\wo}{w_{O}}
\renewcommand{\P}{\mathcal{P}}
\renewcommand{\H}{H}

\newcommand{\A}{\mathcal{A}}
\newcommand{\U}{\mathcal{U}}
\newcommand{\True}{\mathtt T}
\newcommand{\False}{\mathtt F}
\newcommand{\spec}{\varphi}

\newcommand{\LTL}{LTL\xspace}

\newcommand{\stam}[1]{}
\newcommand{\zug}[1]{\langle#1\rangle}
\newcommand{\set}[1]{{\{#1\}}}

\newcommand{\prop}[2][]{{\it prop}_{#1}(#2)}
\newcommand{\cl}[1]{{\it cl}_{#1}}

\newcommand{\D}{D}
\newcommand{\V}{V}
\newcommand{\C}{C}
\newcommand{\G}{G}
\newcommand{\vhcg}{(\V,\H,\C,\G)}
\newcommand{\GL}{\mathbf{G}}
\newcommand{\UN}{\mathbf{U}}
\newcommand{\X}{\mathbf{X}}
\newcommand{\FU}{\mathbf{F}}
\renewcommand{\iff}{\leftrightarrow}

\newcommand{\req}{{\it req}}
\newcommand{\sens}{{\it sens}}
\newcommand{\open}{{\it open}}
\newcommand{\high}{{\it high}}

\newcommand{\ctls}{CTL$^\star$\xspace}

\newtheorem{theorem}{Theorem}[section]
\newtheorem{proposition}[theorem]{Proposition}

\newtheorem{remark}{Remark}[section]
\newtheorem{xmpl}[theorem]{Example}
\newenvironment{example}{\begin{xmpl}\rm}{ \end{xmpl}}

\def\squarebox#1{\hbox to #1{\hfill\vbox to #1{\vfill}}}
\newcommand{\qed}{\hspace*{\fill}
 	\vbox{\hrule\hbox{\vrule\squarebox{.667em}\vrule}\hrule}\smallskip}
 \newenvironment{proof}{\begin{trivlist}
 		\item[\hspace{\labelsep}{\bf\noindent Proof: }]
 	}{\qed\end{trivlist}}

\title{Synthesis with Guided Environments\thanks{A preliminary version was published in {\em Tools and Algorithms for the Construction and Analysis of Systems - 31st International Conference}, 2025. Supported in part by the Israel Science Foundation, Grant 2357/19, and the European Research Council, Advanced Grant ADVANSYNT.}% <-this % stops a space
}
\author{Orna Kupferman and Ofer Leshkowitz\\
School of Computer Science and Engineering\\
The Hebrew University of Jerusalem, Israel} 
\begin{document}
\maketitle

\begin{abstract}
In the synthesis problem, we are given a specification, and we automatically generate a system that satisfies the specification in all environments. We introduce and study {\em synthesis with guided environments} (SGE, for short), where the system may harness the knowledge and computational power of the environment during the interaction. The underlying idea in SGE is that in many settings, in particular when the system serves or directs the environment, it is of the environment's interest that the specification is satisfied, and it would follow the guidance of the system. Thus, while the environment is still
hostile, in the sense that the system should satisfy the specification no matter how the environment assigns values to the input signals, 
in SGE the system assigns values to some output signals and guides the environment via {\em programs\/} how to assign values to other output signals. 
A key issue is that these assignments may depend on input signals that are hidden from the system but are known to the environment, using programs like ``copy the value of the hidden input signal $x$ to the output signal $y$." SGE is thus particularly useful in settings where the system has partial visibility.

We solve the problem of SGE, show its superiority with respect to traditional synthesis, and study theoretical aspects of SGE, like the complexity (memory and domain) of programs used by the system, SGE for branching-time specifications, and the connection of SGE to synthesis of (possibly distributed) systems with partial visibility.

\end{abstract}


\section{Introduction}
\label{intro}
\IEEEPARstart{S}{ynthesis} is the automated construction of a system from its specification \cite{BCJ18}. 
Given a linear temporal logic (\LTL) formula $\spec$ over sets $I$ and $O$ of input and output signals, the goal is to return a {\em transducer} that {\em realizes\/} $\spec$. 
At each moment in time, the transducer reads a truth assignment, generated by the environment, to the signals in $I$, and generates a truth assignment to the signals in $O$. 
% ofer1: This repeats forever --> This process continues indefinitely
This process continues indefinitely, generating an infinite computation in $(2^{I \cup O})^\omega$. The transducer realizes $\spec$ if its interactions with all environments generate computations that satisfy $\spec$~\cite{PR89a}.

The fact the system has to satisfy its specification in all environments has led to a characterization of the environment as {\em hostile}. In particular, in the game that corresponds to synthesis, the objective of the environment is to violate the specification $\varphi$. 
% ofer1: trying to make a distinction between theory and reality.
%While this game-theoretic approach correctly formalizes systems that realize the specification against all environments,
%orna2: some small changes 
In real-life applications, the satisfaction of the specification is often also in the environment's interest. In particular, in cases where the system serves or guides the environment, we can expect the environment to follow guidance from the system. We introduce and study {\em synthesis with guided environments} (SGE, for short), where the system may harness the knowledge and computational power of the environment during the interaction, guiding it how to assign values to some of the output signals.

Specifically, in SGE, the set $O$ of output signals is partitioned into sets $\C$ and $\G$ of {\em controlled\/} and {\em guided\/} signals. Then, a system that is synthesized by SGE assigns values to the signals in $\C$ and guides the environment how to assign values to the signals in $\G$. 
Clearly, not all output signals may be guided to the environment. For example, physical actions of the system, like closing a gate or raising a robot arm, cannot be performed by the environment. In addition, 
%orna2: removed it, as it may cause confusion -- it transits to a setting with hidden outputs
%we may prefer to keep some assignments private, and 
%orna2: some rewrites
it may be the case that the system cannot trust the environment to follow its guidance. 
%it may be the case that we simply cannot trust the environment to follow the instructions of the system for some of the output signals. 
As argued above, however, in many cases we can expect the environment to follow the system's guidance, and 
%orna12 "fairness" is a profetional term, so better to avoid it, so we won't get reviews with references to work in fairness...
require the system to satisfy the specification only when the environment follows its guidance. 
%orna12 new
Indeed, when we go to the cinema and end up seeing a bad movie, we cannot blame a recommendation system that guided us not to choose this movie. The recommendation system is bad (violates its specification, in our synthesis story) only if a user that follows its guidance is disappointed.% 
%it is fair to require the system to satisfy the specification only when the environment obeys its instructions.
\footnote{Readers who are still concerned about harnessing the environment towards the satisfaction of the specification, please note that the environment being hostile highlights the fact that the system has to satisfy the specification for {\em all\/} input sequences. This is still the case also in SGE: While we expect the environment to follow guidance received from the system, we do not limit the input sequences that the environment may generate. 
%orna12 new
Thus, there are no assumptions on the environment or collaboration between the system and the environment in the senses studied in \cite{BEK15,AMNS23}: the setting is as in traditional synthesis, only with assignments to some output signals being replaced by programs that the environment is expected to follow.} 

One advantage of SGE is that it enables a decomposition of the satisfaction task between the system and the environment. We will get back to this point after we describe the setting in more detail. 
%ofer1: Added However
%orna2: good, and I was taught that "howevers" are better inside sentences
The main advantage of SGE, however, has to do with {\em partial visibility}, namely the setting where the set $I$ of input signals is partitioned into sets $V$ and $\H$ of {\em visible\/} and {\em hidden\/} signals, and the systems views only the signals in $V$. Partial visibility makes synthesis much harder, as the hidden signals still appear in the specification, yet the behavior of the system should be independent of their truth values.

Synthesis with partial visibility has been the subject of extensive research.
One line of work studies the technical challenges in solving the problem \cite{Rei84,KV00a,CDHR07}.
Essentially,  in a setting with full visibility, the interaction between the system and the environment induces a single computation and hence a single run of a deterministic automaton for the specification. Partial visibility forces the algorithm to maintain subsets of states of the automaton. 
%orna2: less good, but shorter
Indeed, now the interaction induces several computations, obtained by the different possible assignments to the hidden signals. 
%Indeed, now the interaction induces several computations, obtained by extending the assignments to the visible and output signals by different assignments to the hidden signals. 
This makes synthesis with partial visibility exponentially harder than synthesis with full visibility for specifications given by automata. When the specification is given by an \LTL formula, this exponential price is dominated by the doubly-exponential translation of \LTL formulas to deterministic automata, thus \LTL synthesis with partial visibility is 2EXPTIME-complete, and is not harder than \LTL synthesis with full visibility \cite{Ros92}. 

A second line of work studies different settings in which partial visibility is present. This includes, for example,  {\em distributed systems} \cite{PR90,KV01,Sch08}, systems with controlled {\em sensing} \cite{CMH08,AKK19} or with assumptions about visibility \cite{FMM24}, and systems that maintain {\em privacy} \cite{WRRLS18,KL22}. 
Finally, researchers have studied alternative forms of partial visibility (e.g., {\em perspective games}, where visibility of all signals is restricted in segments of the interaction), as well as partial visibility in richer settings (e.g., {\em multi-agent} \cite{AHK02,BMMRV17,GPW18} or {\em stochastic\/} \cite{KR97,CDNV14} systems) or in problems that are strongly related to synthesis (e.g., control \cite{KS98},  planning \cite{dGV16}, and rational synthesis \cite{BMV17,FGR18}).

To the best of our knowledge, in all settings in which synthesis with partial visibility has been studied so far, there is no attempt to make use of the fact that the assignments to the hidden signals are known to the environment. In SGE, the guidance that the system gives the environment may refer to the values of the hidden signals. Thus, the outcome of SGE is a {\em transducer with a guided environment\/} (TGE): a transducer whose transitions depend only on the assignments to the signals in $V$, it assigns values to the signals in $\C$, and guides the environment how to assign values to the signals in $\G$ using programs that may refer to the values of the signals in $\H$.

Consider, for example, a system in a medical clinic that directs patients  ``if you come for a vaccine, go to the first floor; if you come to the pharmacy, go to the second floor''. Such a system is as correct as a system that asks customers for the purpose of their visit and then outputs the floor to which they should go. 
%ofer2: the first type --> this type
%Clearly, if the customers prefer not to reveal the purpose of their visit, we can direct them only with a system of the first type. 
Clearly, if the customers prefer not to reveal the purpose of their visit, we can direct them only with a system of this type.
This is exactly what SGE  does: it replaces assignments that depend on hidden signals or are complicated to compute with instructions to the environment. As another example, consider a smart-home controller that manages various smart devices within a home by getting inputs from devices like thermostats and security cameras, and generating outputs to devices like lighting systems or smart locks. When it is desirable to hide from the controller information like sleep patterns or number of occupants, we can do that by limiting its input, and guiding the user about the activation of some output devices, for example with instructions like, ``if you expect guests, unlock the backyard gate". In Examples~\ref{app ex1} and~\ref{app ex2}, we describe more elaborated examples, in particular of a server that directs users who want to upload data to a cloud. The users may hide from the server the sensitivity level of their data, and we can expect them to follow instructions that the server issues, for example instructions to use storage of high security only when the data they upload is sensitive.

We study several aspects of SGE.
We start by examining the memory used by the environment. Clearly, this memory can be used to reduce the state space of the TGE. To see this, note that in an extreme setting in which the TGE can guide all output signals, it can simply instruct the environment to execute a transducer that realizes the specification. Beyond a trade-off between the size of the TGE and the memory of the environment, we discuss how the size of the memory depends on the sets of visible and guided signals. For example, we show that, surprisingly, having more guided signals does not require more memory, yet having fewer guided signals may require more memory. 
% ofer1: only doubly-exponential upper bound (not tight)
% ofer1: This paragraph is about section 3. The double exponential upper bound comes later in section 4 
%Our most technically-challenging result in this front is a doubly-exponential upper bound on the size of the memory needed by the environment, which leads to an overall triply-exponential upper bound for the SGE problem with unbounded environment memory. 

We then describe an automata-based solution for SGE.
The main challenge in SGE is as follows. Consider a system that views only input signals in $V$. Its interactions with environments that agree on the signals in $V$ generate the same response. In synthesis with partial visibility, this is handled by {\em universal tree automata\/} that run on trees with directions in $2^V$, thus trees whose branches correspond to the interaction from the point of view of the system \cite{KV00a}. 
%% ofer1: making it clear that its G that is about to be depended on H, not the programs.
In the setting of SGEs, the interactions with different environments that agree on the signals in $V$ still generate the same response, but this response now involves programs that guide the environment on how to assign values to signals in $\G$ based on the values assigned to the signals in $\H$. As a result, the computations induced by these interactions may differ (not only on $\H$ but also on $\G$).
%ofer1: some rewrites
%orna2: Combines with "this is technically hard"
This differences between SGE and traditional synthesis with partial visibility is our main technical challenge. We show that we can still reduce SGE to the nonemptiness problem of universal co-B\"uchi tree automata, proving that SGE for \LTL specifications is 2EXPTIME-complete. In more detail, given an \LTL formula $\spec$ and a bound $k$ on the memory that the environment may use, the constructed automaton is of size exponential in $|\spec|$ and linear in $k$, making SGE doubly-exponential in $|\spec|$ and exponential in $k$.
%ofer1: added the doubly-exponential upper bound to here. 
 Finally, we prove a doubly-exponential upper bound on the size of the memory needed by the environment, leading to an overall triply-exponential upper bound for the SGE problem with unbounded environment memory.

We continue and study the domain of programs that TGEs use. Recall that these programs guides the environment how to update its memory and assign values to the guided signals, given a current memory state and the current assignment to the hidden signals. Thus, for a memory state space $M$, each program is of the form $p:M \times 2^\H \rightarrow M \times 2^\G$. We study ways to reduce the domain $2^{\H}$, which is the most dominant factor. The reduction depends on the specification $\spec$ we wish to synthesize. We argue that the SGE algorithm can restrict attention to {\em tight\/}  programs: ones whose domain is a set of predicates over $\H$ obtained by simplifying propositional sub-formulas of $\spec$. Further simplification is achieved by exploiting the fact that programs are called after an assignment to the signals in $V \cup \C$ has been fixed, and exploiting dependencies among all signals. 

We next compare our solution with one that views a TGE as two distributed processes that are executed together in a pipeline architecture: the TGE itself, and a transducer with state space $M$ that implements the instructions of the TGE to the environment. We argue that the approach we take is preferable and can lead to a quadratic saving in their joint state spaces, similar to the saving obtained by defining a regular language as the intersection of two automata. 
%ofer3: not sure I understand the following or why its relevant
Generating programs that manage the environment's memory efficiently is another technical challenge in SGE.

Finally, we extend the SGE framework to handle branching-time specifications, such as those expressed in \ctls. In this setting, the interaction between the system and its environment induces a {\em  computation tree\/} whose branches correspond to different behaviors of the environment. Indeed, the system is deterministic, and nondeterminism is present due to the different possible assignments the environment may provide to the visible and hidden input signals \cite{KB17,BSK17}. The branching-time setting enables the synthesis of properties like ``in all computations, the environment can always assign values to the input signals so that the output signal {\it release} is eventually true". Note that such a behavior cannot be specified in LTL. Synthesis of branching-time specifications with hidden inputs is of particular interest and challenge, as the branching degree of the computation tree viewed by the system is induced only by the visible signals, and it thus smaller than the actual branching degree \cite{KV00a}.  In SGEs, the system may guide the environment to assign different values to the output signals in branches that agree on the visible inputs. We show that our automata-theoretic approach naturally generalizes to \ctls specifications by constructing alternating parity tree automata that operate over the computation tree induced by a TGE. This allows us to synthesize systems that satisfy \ctls specifications in the presence of guided environments, with a complexity that remains 2EXPTIME-complete with respect to the specification size.
%orna12: repeats
\stam{
From a technical point of view, SGE introduces several challenges on top of these addressed in work on traditional synthesis. One primary challenge is managing the interplay between four distinct types of signals including hidden input signals and guided output signals. Although the setting resembles synthesis with partial visibility, SCE introduces an additional layer of complexity: the system can harness hidden information and allow it to influence the values assigned to guided signals through programs that instruct the environment. Technically, it means that computations that are indistinguishable in settings with partial visibility become distinguishable in SCE: indeed the responses of the system should coincide on such computations, but identical responses may lead to different computations. In addition, the system must generate programs that manage the environment's memory efficiently, and ensure a correct behaivor against all assignments to the hidden inputs. These requirements make the problem inherently more complex than standard LTL synthesis. In particular, the complexity of the LTL SGE problem when the environment's memory is unbounded remains open, with a current 3EXPTIME upper bound. This upper bound can also be obtained by viewing a TGE as a distributed system with a pipeline architecture of two layers, highlighting its close relation to distributed systems, a relationship we explore in a dedicated section.
}

We conclude with directions for future research. Beyond extensions of the many settings in which synthesis has been studied to a setting with guided environments, we discuss two directions that are more related to ``the guided paradigm'' itself: settings with {\em dynamic\/} hiding and guidance of signals, thus when $\H$ and $\G$ are not fixed throughout the interaction; and {\em bounded SGE}, where, as in traditional synthesis \cite{SF07,KLVY11}, beyond a bound on the memory used by the environment, there are bounds on the size of the state space of the TGE and possibly also on the size of the state space of its environment. 

\section{Preliminaries}\label{sec prelim}

%\subsection{LTL}
We describe on-going behaviors of reactive systems using the linear temporal logic \LTL \cite{Pnu81}. 
We consider systems that interact via sets $I$ and $O$ of  input and output signals, respectively.  
Formulas of \LTL are defined over $I \cup O$ using the usual
Boolean operators and the temporal operators $\GL$ (``always") and $\FU$ (``eventually"), $\X$ (``next time'') and
$\UN$ (``until''). The semantics of \LTL is defined with respect to infinite 
computations in $(2^{I \cup O})^\omega$. Thus, each \LTL formula $\spec$ over $I \cup O$ induces a language $L_\spec\subseteq (2^{I \cup O})^\omega$ of all computations that satisfy $\spec$. 

The {\em length\/} of an \LTL formula $\spec$, denoted $|\spec|$, is the number of nodes in the generating tree of $\spec$. Note that $|\spec|$ bounds the number of sub-formulas of $\spec$.

%\subsection{Transducers with Cooperative Environments}
We model reactive systems that interact with their environments by finite-state transducers. 
A \emph{transducer} is a tuple $\T=\zug{I,O,S,s_0,\delta,\tau}$, where $I$ and $O$ are sets of input and output signals, $S$ is a finite set of states, $s_0\in S$ is an initial state, $\delta:S\times 2^I\rightarrow S$ is a transition function, and $\tau: S \times 2^I \rightarrow 2^O$ is a function that labels each transition by an assignment to the output signals. 
%Both the transition function $\delta$ and the labeling function $\tau$ can be extended to non-empty finite words $x_V\in (2^V)^+$, $\delta^+:S\times (2^V)^+\rightarrow S$ and $\tau^+:S\times (2^V)^+\rightarrow 2^O$ respectively, by $\delta^+(s,v)=\delta(s,v)$ and $\tau^+(s,v)=\tau(s,v)$ for all $v\in 2^V$, and $\delta^+(s,x_V\cdot v)=\delta(\delta^+(s,x_V),v)$ and $\tau^+(s,x_V\cdot v)=\tau(\delta^+(s,x_V),v)$ for all $x_V\in (2^V)^+$ and $v\in 2^V$. 
Given an infinite sequence $\wi=i_1\cdot i_2\cdots\in (2^I)^{\omega}$ of assignments to input signals, $\T$ generates an infinite sequence $\wo=o_1\cdot o_2\cdots\in (2^O)^{\omega}$ of assignments to output signals. 
Formally, a {\em run\/} of $\T$ on $\wi$ is an infinite sequence of states $s_0\cdot s_1\cdot s_2\cdots$, where for all $j \geq 1$, we have that $s_j=\delta(s_{j-1},i_j)$. 
%Or equivalently, $\delta^+(s_0,i_1\cdots i_j)$. 
Then, the sequence $w_O$ is obtained from the assignments along the transitions that the run traverses. Thus for all $j \geq 1$, we have that $o_j=\tau(s_{j-1},i_j)$. %, or simply $o_j=\tau^+(s_0,i_1\cdots i_j)$. 
We define the \emph{computation\/} of $\T$ on $\wi$ to be the word $\T(\wi)=(i_1\cup o_1)\cdot (i_2\cup o_2)\cdots \in (2^{I\cup O})^{\omega}$. %The \emph{language} of $\T$, denoted $L(\T)$, is the set of its computations; that is $L(\T)=\set{\T(\wi) : \wi\in (2^I)^{\omega}}$.

For a specification language $L_\spec \subseteq (2^{I\cup O})^{\omega}$, we say that $\T$ {\em $(I,O)$-realizes\/} $L_\spec$ if for every input sequence $\wi\in (2^I)^{\omega}$, we have that $\T(\wi) \in L_\spec$. 
%Equivalently, if $L(\T) \subseteq L_\spec$. 
In the {\em synthesis\/} problem, we are given a specification language $L_\spec$ and a partition of the signals to $I$ and $O$, and we have to return a transducer that $(I,O)$-realizes $L_\spec$ (or determine that $L_\spec$ is not realizable).  The language $L_\spec$ is typically given by an LTL formula $\spec$. We then talk about realizability or synthesis of $\spec$ (rather than $L_\spec$).

%\begin{theorem}\label{ltl syn 2exp}
%	The LTL synthesis problem is 2EXPTIME-complete~\cite{?}.
%\end{theorem}

In synthesis with {\em partial visibility}, we seek a system that satisfies a given specification in all environments even when it cannot observe the assignments to some of the input signals. Formally, the set of input signals is partitioned into {\em visible\/} and {\em hidden\/} signals, thus $I=\V \cup \H$. The specification $\spec$ is still over $\V \cup \H \cup O$, yet the behavior of the transducer that models the system is independent of $\H$. Formally, $\T=\zug{\V,\H,O,S,s_0,\delta,\tau}$, where now $\delta:S\times 2^\V\rightarrow S$ and $\tau: S \times 2^\V \rightarrow 2^O$.
%ofer2: fixed the definition of the outputs o_j to depend on V and not H - we are still in the standard setting of incomplete information.
Given an infinite sequence $\wi=(v_1 \cup h_1)\cdot (v_2 \cup h_2) \cdots\in (2^{\V \cup \H})^{\omega}$, the run and computation of $\T$ on $w_I$ is defined as in the case of full visibility, except that now, for all $j \geq 1$, we have that $s_j=\delta(s_{j-1},v_j)$ and $o_j=\tau(s_{j-1},v_j)$.

We can now define a \emph{transducer with a guided environment} (TGE for short). TGEs extend traditional transducers by instructing the environment how to manage its guided signals. A TGE may be executed in a setting with partial visibility, thus $I = \V \cup \H$. 
It uses the fact that the assignments to the signals in $\H$ are known to the environment, and it guides the assignment to some of the output signals to the environment. 
As discussed in Section~\ref{intro}, in practice not all output signals can be guided. 
%ofer1: these examples already appear in the intro. maybe we can remove to save space
%orna2: indeed. Added a reference to the intro above instead.
%In practice, not all output signals can be guided. 
%For example, physical actions of the system, like closing a gate or raising a robot arm, cannot be performed by the environment. 
%
%In addition, some assignments should better stay private, and it may be the case that the environment cannot be trusted to follow the instructions of the system. 
Formally, the set of output signals is partitioned into {\em controlled\/} and {\em guided\/} signals, thus $O=\C \cup \G$. In each transition, the TGE assigns values to the signals in $\C$ and instructs the environment how to assign values to the signals in $\G$. The environment may have a finite memory, in which case the transducer also instructs the environment how to update the memory in each transition. The instructions that the transducer generates are represented by {\em programs}, defined below.

Consider a finite set $M$ of memories, and sets $\H \subseteq I$ and $\G \subseteq O$ of input and output signals. Let $\P_{M,\H,\G} = (M \times 2^\H\rightarrow M\times 2^\G)$ denote the set of \emph{propositional programs} that update the memory state and assign values to signals in $\G$, given a memory in $M$ and an assignment to the signals in $\H$.
Note that each member of $\P_{M,\H,\G}$ is of the form $p:M \times 2^\H \rightarrow M\times 2^\G$.
For a program $p\in \P_{M,\H,\G}$, let $p_M:M\times 2^\H\rightarrow M$ and $p_\G:M\times 2^\H\rightarrow 2^\G$ be the projections of $p$ onto $M$ and $\G$ respectively. Thus, $p(m,h)=\zug{p_M(m,h),p_\G(m,h)}$ for all $\zug{m,h}\in M\times 2^\H$. 
In Section~\ref{sec programs} we discuss ways to restrict the set of programs that a TGE may suggest to its environment without affecting the outcome of the synthesis procedure. 

Now, a {\em TGE\/} is $\T=\zug{\V,\H,\C,\G,S,s_0,\delta,M,m_0,\tau}$, where $\V$ and $\H$ are sets of visible and hidden input signals, $\C$ and $\G$ are sets of controlled and guided output signals, $S$ is a finite set of states, $s_0 \in S$ is an initial state, $\delta:S\times 2^\V\rightarrow S$ is a transition function, $M$ is a set of memories that the environment may use, $m_0 \in M$ is an initial memory, and $\tau:S\times 2^\V \rightarrow 2^\C \times \P_{M,\H,\G}$ labels each transition by an assignment to the controlled output signals and a program.
Note that $\delta$ and $\tau$ are independent of the signals in $\H$, and that $\tau$ assigns values to the signals in $\C$ and instructs the environment how to use the signals in $\H$ in order to assign values to the signals in $\G$. The latter  reflects the fact that the environment does view the signals in $\H$, and constitute the main advantage of TGEs over standard transducers. 

Given an infinite sequence $\wi=(v_1 \cup h_1)\cdot (v_2 \cup h_2) \cdots\in (2^{\V \cup \H})^{\omega}$,
the run of $\T$ on $\wi$ is obtained by applying $\delta$ on the restriction of $\wi$ to $\V$. Thus, $r=s_0\cdot s_1\cdot s_2\cdots$, where for all $j \geq 1$, we have that $s_{j}=\delta(s_{j-1},v_j)$.
The interaction of $\T$ with the environment generates an infinite sequence $w_C=c_1 \cdot c_2 \cdot c_3 \cdots \in (2^\C)^\omega$ of assignments to the controlled signals, an infinite sequence $w_M=m_0\cdot m_1\cdot m_2 \cdots \in M^\omega$ of memories, and an infinite sequence $w_P = p_1\cdot p_2\cdot p_3\cdots  \in (\P_{M,\H,\G})^\omega$ of programs, which in turn generates an infinite sequence $w_G=d_1 \cdot d_2 \cdot d_3 \cdots \in (2^\G)^\omega$ of assignments to the guided signals. Formally, for all $j \geq 1$, we have that $\zug{c_j,p_j} = \tau(s_{j-1},v_j)$ and $\zug{m_j,d_j}=p_j(m_{j-1},h_j)$.
The computation of $\T$ on $\wi$ is then $\T(\wi)=(v_1 \cup h_1 \cup c_1\cup d_1)\cdot (v_2 \cup h_2 \cup c_2\cup d_2) \cdots \in (2^{\V \cup \H \cup C \cup \G})^{\omega}$.
%As with traditional transducers, the language of a TGE $\T$ is $L(\T)=\set{\T(w_I):w_I\in (2^I)^\omega}$.
 
Note that while the domain of the programs in $\P_{M,\H,\G}$ is $M \times 2^\H$, programs are chosen in $\T$ along transitions that depend on $2^\V$. Thus, effectively, assignments made by programs depend on signals in both $\H$ and $\V$.

For a specification language $L_\spec \subseteq (2^{I\cup O})^{\omega}$, partitions $I=\V\cup \H$ and $O=\C\cup \G$, and a bound $k \geq 1$ on the memory that the environment may use, we say that a TGE $\T$, {\em $\vhcg$-realizes $L_\spec$ with memory $k$}, if $\T=\zug{\V,\H,C,\G,S,s_0,\delta,M,m_0,\tau}$, with $|M|=k$, and 
%ofer: Changing to L(\T)
%orna1
$\T(\wi) \in L_\spec$, for every input sequence $\wi\in (2^I)^{\omega}$.
%$L(\T)\subseteq L_\spec$.
%ofer2: made into a single paragraph and slighted rephrased to make it less repetitive
Then, in the {\em synthesis with guided environment\/} problem (SGE, for short), given such a specification $L_\spec \subseteq (2^{I \cup O})^\omega$, a bound $k\geq 1$ and partitions $I=\V\cup \H$ and $O=\C\cup \G$, we should construct a TGE with memory $k$ that $\vhcg$-realizes $L_\spec$ (or determine that no such TGE exists). 
 
Let us consider the special case of TGEs that operate in an environment with no memory, thus $M=\{m\}$, for a single memory state $m$.
Consider first the setting where $\T$ has full visibility, thus $\V=I$ and $\H=\emptyset$. Then, programs are of the form $p:\set{m}\times 2^{\emptyset}\rightarrow \set{m}\times 2^\G$, and so each program is a fixed assignment to the signals in $\G$. Hence, TGEs that have full visibility and operate in an environment with no memory coincide with traditional transducers. 
Consider now a setting with partial visibility, thus $\V \neq I$. Now, programs are of the form $p:\set{m}\times 2^{\H}\rightarrow \set{m}\times 2^\G$, allowing the TGE to guide the environment in a way that depends on signals in $\H$. 

In Examples~\ref{app ex1} and~\ref{app ex2} below, we demonstrate how collaboration from the environment can lead to the realizability of specifications that are otherwise non-realizable. 

\begin{example}\label{app ex1}
{\bf [TGEs for simple specifications]}
Let $I=\{i\}$, $O=\{o\}$, and consider the simple specification $\spec_1=\GL (i \leftrightarrow o)$. 
Clearly, $\spec_1$ is realizable in a setting with full visibility, yet is not realizable in a setting with partial visibility and $\H=\{i\}$. A TGE can realize $\spec_1$ even in the latter setting. Indeed, even when the environment has no memory, it can delegate the assignment to $o$ to the environment and guide it to copy the value of $i$ into $o$. Formally, $\spec_1$ is realizable by the $\T_1=\zug{\emptyset,\{i\},\emptyset,\{o\},\{s\},\{s\},\delta,\{m\},m,\tau}$, with $\delta(s,\emptyset)=s$ and $\tau(s,\emptyset)=\zug{\emptyset,p}$, where the program $p \in \P_{\{m\},\{i\},\{o\}}$ is the command $o := i$. 

Consider now the specification $\spec_2=\GL(i\leftrightarrow \X o)$. Here, a TGE in a setting in which $i$ is not visible needs an environment with at least one register, inducing a memory of size $2$. Then, the TGE can instract the environment to store the value of $i$ in the  register, and use the stored value when assigning a value to the signal $o$ in the next round. Formally, $\spec_2$ is realizable by the TGE $\T_2=\zug{\emptyset,\{i\},\emptyset,\{o\},\{s\},\{s\},\delta,\{m_0,m_1\},m_0,\tau}$,with $\delta(s,\emptyset)=s$ and $\tau(s,\emptyset)=\zug{\emptyset,p}$, where $p \in \P_{\{m_0,m_1\},\{i\},\{o\}}$ instructs the environment to move to $m_0$ when $i=\False$ and to $m_1$ when $i=\True$, and to assigns $\False$ to $o$ when it is in $m_0$ and $\True$ when it is in $m_1$.
\hfill \qed \end{example}

\begin{example}\label{app ex2}
{\bf [A TGE implementing a server]}
Consider a server that directs users who want to upload data to a cloud.
The set of input signals is $I=\{\req,\sens\}$, where $\req$ holds when a user requests to upload data, and $\sens$ holds when the data is sensitive. The set of output signals is $O=\{\open,\high\}$, where $\open$ holds when the cloud is open for uploading, and $\high$ holds when the storage space is of high security.
Users pay more for storage of high security. Therefore, in addition to guaranteeing that all requests are eventually responded by $\open$, we want the server to direct the users to use storage of high security only when the data they upload is sensitive. In addition, the cloud cannot stay always open to uploads. Formally, we want to synthesize a transducer that realizes the conjunction $\spec$ of the following \LTL formulas. 
\begin{itemize}
	\item
	$\spec_1=\GL((\req \wedge \sens) \rightarrow ((\neg \open) \UN (\open \wedge \high)))$.
	\item
	$\spec_2=\GL((\req \wedge \neg \sens) \rightarrow ((\neg \open) \UN ((\open \wedge \neg \high) \vee (\req \wedge \sens)))$.
	\item
	$\spec_3=\GL\FU \neg \open$.
\end{itemize}

A server that has full visibility of $I$ can realize $\spec$. To see this, note
that a server can open the cloud for uploading whenever a request arrives, storing the data in a storage of high security  iff it is sensitive. Then, in order to satisfy $\spec_3$, the server should delay the response to successive requests, but this does not prevent it from satisfying $\spec_1$ and $\spec_2$.  

Users may prefer not to share with the server information about the sensitivity of their data.
In current settings of synthesis with partial visibility, this renders $\spec$ to be unrealizable.
Indeed, when the sensitivity of the data is hidden from the scheduler, thus when $\V=\{\req\}$ and $\H=\{\sens\}$, the behavior of the server is independent of $\sens$, making it impossible to assign values to $\high$ in a way that satisfies $\spec$ in all environments.

While the output signal $\open$ controls access to the cloud, the
output signal $\high$ only directs the user which type of storage to use, and it is of the user's interest to store her data in a storage of an appropriate security level. Accordingly, the environment can be guided with the assignment to $\high$. Thus, $\C=\{\open\}$ and $\G=\{\high\}$.

\begin{figure}[ht]
	\centering
	\includegraphics[width=0.48\textwidth]{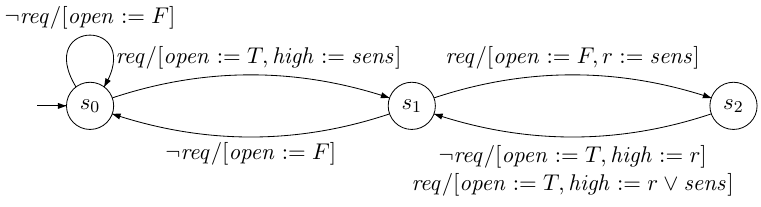}
	\caption{A transducer with a guided environment.}
	\label{fig storage}
\end{figure}

In Figure~\ref{fig storage}, we present a TGE that implements such a guidance and realizes $\spec$. 
Each transition of the TGE is labeled by a guarded command of the form $v / [c,p]$, where $v \in 2^\V$ is an assignment to the visible input signals (or a predicate describing several assignments), $c \in 2^\C$ is an assignment to the controlled output signals, and $p$ is a program that instructs the environment how to assign values to the guided signals and how to update its memory.\footnote{For clarity, in the TGE in Figure~\ref{fig storage},  the assignment to $\high$ is missing in transitions with $\open=\False$, where $\high$ can be assigned arbitrarily, and the assignment to the register $r$ is missing in transitions that are taken when there are no pending requests, in which case $r$ can be assigned arbitrarily.}
The environment uses a memory with one register $r$. Accordingly, the transitions of the TGE are guarded by the truth value of $\req$, they include an assignment to $\open$, and then a program that instructs the environment how to assign a value to $\high$ and to the register $r$. The key advantage of TGEs is that these instructions may depend on the hidden input signals. Indeed, the environment does know their values. For example, when the TGE moves from $s_1$ to $s_2$, it instructs the environment to assign to $r$ the value of $\sens$. Then, when the TGE moves from $s_2$ to $s_1$, it instructs the environment to use the value stored in $r$ (as well as the current value of $\sens$) when it assigns a value to $\high$. 
\hfill \qed \end{example}

\section{On the Memory Used by the Environment}

In this section we examine different aspects of the memory used by the environment. We discuss a trade-off between the size of the memory and the size of the TGE,  
and how the two depend on the partitions $I=\V\cup \H$ and $O=\C\cup \G$ of the input and output signals. 

\begin{example}
\label{illustrate memory}
In order to better understand the need for memory and the technical challenges around it, consider the tree appearing in Figure~\ref{fig:memorytree}.

\begin{figure}[ht]
	\centering
		\includegraphics[width=0.30\textwidth]{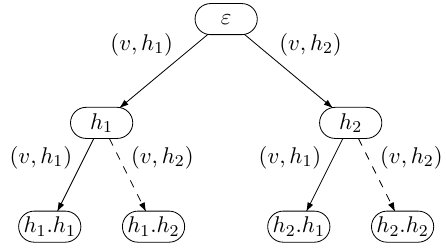} 
		\caption{Two rounds of interaction with a TGE $\T$.}
		\label{fig:memorytree}
\end{figure}

	The tree describes two rounds of an interaction between a TGE $\T$ and its environment. In both rounds, $\T$ views the assignment $v \in 2^\V$. Consider two assignments $h_1,h_2 \in 2^\H$. The tree presents the four possible extensions of the interaction by assignments in $\{h_1,h_2\}$. Since the signals in $\H$ are hidden, $\T$ cannot distinguish among the different branches of the tree. In particular, it has to suggest to the environment the same program in all the transitions from the root to its successors, and the same program in all the transitions from these successors to the leaves. While these programs may depend on $\H$, a program with no memory would issue the same assignment to the signals in $\G$ along the two dashed transitions $\zug{h_1,h_1.h_2}$ and $\zug{h_2,h_2.h_2}$.
	Indeed, in both transitions, the environment assigns $h_2 \in 2^\H$ to the hidden signals.  Using memory, $\T$ can instruct the environment to maintain information about the history of the computation, and thus distinguish between computations with the same visible component and the same current assignment to the hidden signals. For example, by moving to memory state $m_1$ with $h_1$ and to memory state $m_2$ with $h_2$, a program may instruct the environment to assign different values along the dashed lines.  
\hfill \qed \end{example}

We first observe that in settings in which all output signals can be guided, the TGE can instruct the environment to simulate a transducer that realizes the specification. Also, in settings with full visibility, the advantage of guided environments is only computational, thus specifications do not become realizable, yet may be realizable by smaller systems. Formally, we have the following. 

\begin{theorem}
\label{trade offs}
	Consider an LTL specification $\spec$ over $I \cup O$.
	\begin{enumerate}
		\item
		If $\spec$ is realizable by a transducer with $n$ states, then $\spec$ is $(\emptyset,I,\emptyset,O)$-realizable by a one-state TGE with memory $n$.
		\item
		For every partition $I=\V \cup \H$ and $O=\C \cup \G$, if $\spec$ is $(\V,\H,C,\G)$-realizable by a TGE with $n$ states and memory $k$, then $\spec$ is $(I,O)$-realizable (in a setting with full visibility) by a transducer with $n \cdot k$ states. In particular, if $\spec$ is $(\emptyset,I,\emptyset,O)$-realizable by a one-state TGE with memory $k$, then $\spec$ is $(I,O)$-realizable (in a setting with full visibility) by a transducer with $k$ states.
	\end{enumerate}
\end{theorem}

\begin{proof}
	We start with the first claim.
	Given a transducer $\T=\zug{I,O,S,s_0,\delta,\tau}$, we define a program $p \in \P_{S,I,O}$ that instructs the environment to simulate $\T$.
	Formally, for every $s \in S$ and $i \in 2^I$, we have that $p(s,i)=\zug{\delta(s,i),\tau(s,i)}$.
	Consider the TGE $\T'=\zug{\emptyset,I,\emptyset,O,\{s\},s,\delta',S,s_0,\tau'}$, with $\delta'(s,\emptyset)=s$ and $\tau'(s,\emptyset)=\zug{\emptyset,p}$. It is easy to see that for every $w_I\in (2^I)^\omega$, we have that $\T(w_I)=\T'(w_I)$.  In particular, if $\T$ realizes $\spec$, then so does $\T'$.
	
	We continue to the second claim. Note that the special case where $\V=C=\emptyset$ is the ``only if'' direction of the first claim. Given a TGE $\T=\langle \V,H,C,\G$, $S,s_0$, 
	$\delta,M$, $m_0,\tau\rangle$, consider the transducer $\T'=\zug{H \cup \V, \C\cup \G,S \times M, \zug{s_0,m_0},\delta',\tau'}$, where for all $\zug{s,m} \in S \times M$ and $i \in 2^{\H \cup \V}$ with $\tau(s,i \cap \H)=\zug{c,p}$, we have that $\delta'(\zug{s,m},i)=\zug{\delta(s,i \cap \V),p_{M}(i \cap \H)}$ and $\tau'(\zug{s,m},i) = c \cup p_{\G}(i \cap \H)$. It is easy to see that for every $w_I\in (2^I)^\omega$, we have that $\T(w_I)=\T'(w_I)$.  In particular, if $\T$ realizes $\spec$, then so does $\T'$.
\end{proof}

Note that Theorem~\ref{trade offs} also implies that an optimal balance between controlled and guided signals in a setting with full visibility can achieve at most a quadratic saving in the combined states space of the TGE and its memory. 
This is similar to the quadratic saving in defining a regular language by the intersection of two automata. Indeed, handling of two independent specifications can be partitioned between the TGE and the environment. Formally, we have the following. 

%ofer2: above-->over. Should it be "over" or "above"?
\begin{theorem}\label{tce comp quad}
	For all $\,k\geq 1$, there exists a specification $\spec_k$ over $I=\set{i_1,i_2}$ and $O=\set{o_1,o_2}$, such that $\spec_k$ is $(\set{i_1},\set{i_2},\set{o_1},\set{o_2})$-realizable by a TGE with $k$ states and memory $k$, yet a transducer that $(I,O)$-realizes $\spec_k$ needs at least $k^2$ states.
\end{theorem}

\begin{proof}	
	For a proposition $p$, let $\FU^1 p = \FU p$, and for $j \geq 1$, let $\FU^{j+1} p= \FU(p \wedge \X(\FU^j p))$. Thus, $\FU^j p$ holds in a computation $\pi$ iff $p$ holds in at least $j$ positions in $\pi$. 

We define $\spec_k=\spec^1_k\wedge \spec^2_k$, where for $l \in \{1,2\}$, we have that $\spec^l_k=(\FU^k i_l)\iff \FU o_l$. That is, $\spec^l_k$ is over $\set{i_l,o_l}$, and it holds in a computation $\pi$ if $o_l$ is always $\False$ in $\pi$, unless there are $k$ positions in $\pi$ in which $i_l$ is $\True$, in which case $o_l$ has to be $\True$ in at least one position. 

We prove that an $(I,O)$-transducer that realizes $\spec_k$ needs $k^2$ states, yet there is a TGE that $(\set{i_1},\set{i_2},\set{o_1},\set{o_2})$-realizes $\spec_k$ with state space and memory set both of size $k$.

It is easy to see that $\spec_k$ is $(\set{i_l},\set{o_l})$-realizable by a transducer $\T_l^k$ with $k$ states $\{s_0,\ldots,s_{k-1}\}$. 
The state space serves like a counter. Transitions from state $s_j$, for $0\leq j< k-1$, assign $\False$ to $o_l$, looping when $i_1$ is $\False$ and moving to $s_{j+1}$ when $i_1$ is $\True$. From state $s_{k-1}$, the transducer always loops, copying the value of $i_l$ to $o_l$. 

 A TGE that $(\set{i_1},\set{i_2},\set{o_1},\set{o_2})$-realizes $\spec_k$ can have the state space of $\T_1^k$, and label its transitions by a program $p$ that instructs the environment to simulates $\T_2^k$, with memory that corresponds to the state space of $\T_2^k$. The details are similar to the simulation described in the proof of Theorem~\ref{trade offs}~(1). 
 
 On the other hand,  a transducer that $(I,O)$-realizes $\spec_k$ needs to maintain both a counter for $i_1$ and a counter for $i_2$, and so it needs at least $k^2$ states. The formal proof is similar to the proof that an automaton 
 for $F^k i_1 \wedge F^k i_2$ needs $k^2$ states. Essentially, a transducer with fewer states would reach the same state $s$ after reading two input sequences that differ in the number of occurrences of $i_1$ or $i_2$, and would then err on the output of at least one continuation of these sequences. 
\end{proof}

	Unsurprisingly, larger memory of the environment enables TGEs to synthesize more specifications. Formally, we have the following.

\begin{theorem}\label{monotone mem}
	For every $k \geq 1$, the LTL specification $\spec_k=\GL(i \leftrightarrow \X^{k}o)$ is $(\emptyset,\{i\},\emptyset,\{o\})$-realizable by a TGE with memory $2^k$, and is not with memory $2^{k-1}$.
\end{theorem}
	
\begin{proof}
We first describe a TGE with memory that is composed of $k$ registers $\set{r_1,r_2,\ldots,r_k}$, that $(\emptyset,\{i\},\emptyset,\{o\})$-realize $\spec_k$. The TGE has a single state in which it instructs the environment to assign to $o$ the value stored in $r_k$, and to shift the stored values to the right, thus with $r_1$ getting the current value of $i$ and $r_{j+1}$ getting the value stored in $r_j$, for $1 \leq j < k$. This guarantees that the value assigned to $o$ agrees with the value that $i$ was assigned with $k$ rounds earlier, and so $\spec_k$ is satisfied in all environments.

We continue and prove the lower bound. I.e., that a TGE with less than $2^{k-1}$ memory states cannot $(\emptyset,\{i\},\emptyset,\{o\})$-realize $\spec_k$. Consider a TGE $\T$ that $(\emptyset,\{i\},\emptyset,\{o\})$-realizes $\spec$ with memory set $M$. As $\T$ has no visibility, it generates the same sequence $\zug{\emptyset,p^1},\zug{\emptyset,p^2},\zug{\emptyset,p^3},\ldots \in (2^{\{\emptyset\}} \times \P_{M,\{i\},\{o\}})^\omega$, for all input sequences. Consider the function $f:(2^\set{i})^*\rightarrow M$ that maps a finite input sequence $x\in (2^\set{i})^*$ to the memory in $M$ that is reached after the environment generated $x$. Thus, $f(\varepsilon)=m_0$, and for every $x\in (2^\set{i})^*$ and $a\in 2^\set{i}$, we have $f(x\cdot a)=p^{|x|+1}_M(f(x),a)$.

We prove that $f$ is injective on $(2^\set{i})^{k}$. That is, we prove that for every two input sequences $x,y \in (2^\set{i})^{k}$, if $x \neq y$, then $f(x) \neq f(y)$. It will then follow that $|M|\geq 2^{k}$, as required.

Assume by way of contradiction that there are two input sequences $x=x_1\cdot x_2\cdots x_{k-1}\in (2^\set{i})^{k-1}$ and $y=y_1\cdot y_2\cdots y_{k-1}\in (2^\set{i})^{k-1}$, such that $x \neq y$, yet $f(x)=f(y)=m$. 
Since $x \neq y$, there is some position $t\leq k-1$ such that  $x_{t}\neq y_{t}$.

Consider the infinite input sequences $x'=x\cdot \emptyset^\omega$ and $y'=y\cdot \emptyset^\omega$. By our assumption, $\T$ reaches the same memory state $m$ after the first $k-1$ rounds of its interaction with environments that generate $x'$ and $y'$. Since $x'$ and $y'$ agree on their suffix after these rounds, we have that $\T(x')$ and $\T(y')$ agree on the assignment to $o$ on all positions after $k$. In particular, $\T(x')$ and $\T(y')$ agree on the assignment to $o$ in position $k+t$, contradicting the fact that $i$ holds only in one of the assignments $x_{t}$ and $y_t$. Thus, $\T$ does not $(\emptyset,\{i\},\emptyset,\{o\})$-realizes $\spec$, and we have reached a contradiction.
\end{proof}

We turn to consider how changes in the partition of the input signals to visible and hidden signals, and the output signals to controlled and guided ones, 
may affect the required size of the TGE and the memory required to the environment. In \autoref{mem suff}, we show that increasing visibility or decreasing control, does not require more states or memory. On the other hand, in \autoref{mem not suff}, we show that increasing control by the system may require the environment to use more memory. Thus, surprisingly, guiding more signals does not require more memory, yet guiding fewer signals may require more memory. Intuitively, it follows from the can that guiding fewer signals may force the TGE to satisfy the specification in alternative (and more space consuming) ways. 

\begin{theorem}
\label{mem suff}
	Consider an LTL specification $\spec$ that is $(\V,H,C,\G)$-realizable by a TGE with $n$ states and memory $k$.
	\begin{enumerate}
		\item
		For every $I' \subseteq \H$, we have that $\spec$ is $(\V \cup I',\H \setminus I',\C,\G)$-realizable by a TGE with $n$ states and memory $k$.
		\item
		For every $O' \subseteq \C$, we have that $\spec$ is $(\V,\H,C \setminus O',\G \cup O')$-realizable by a TGE with $n$ states and memory $k$.
	\end{enumerate}
\end{theorem}

\begin{proof}
Let $\T=\zug{\V,H,C,\G,S,s_0,\delta,M,m_0,\tau}$ be a TGE that $\vhcg$-realizes $\spec$. We start with the first claim.
Given $I' \subseteq \H$, consider the TGE $\T'=\zug{\V \cup I',\H\setminus I' ,C,G,S,s_0,\delta',M,m_0,\tau'}$, where for every $s \in S$ and $u \in 2^{\V \cup I'}$, we have that $\delta'(s,u)=\delta(s,u \cap \V)$, and if $\tau(s,u \cap \V)=\zug{c,p}$,
then $\tau'(s,u)=\zug{c,p'}$, where for all $g \in 2^{\H \setminus I'}$, we have that $p'(g)=p(g \cup (u \cap I'))$.
It is easy to see that $\T'$ has the same size and memory as $\T$, and that  for every $w_I\in (2^I)^\omega$, we have that $\T(w_I)=\T'(w_I)$.  In particular, if $\T$ realizes $\spec$, then so does $\T'$.

We continue to the second claim. Given $O' \subseteq \C$, consider the TGE $\T'=\zug{\V,\H,\C \setminus O',\G \cup O',S,s_0,\delta,M,m_0,\tau'}$, where for every $s \in S$ and $v \in 2^\V$ with $\tau(s,v)=\zug{c,p}$, we have that $\tau'(s,v)=\zug{c \setminus O',p'}$, where for every $h \in 2^\H$ with $p(h)=\zug{m,d}$, we have $p'(h)=\zug{m,d \cup (c \cap O')}$. It is easy to see that $\T'$ has the same size and memory as $\T$, and that  for every $w_I\in (2^I)^\omega$, we have that $\T(w_I)=\T'(w_I)$.  In particular, if $\T$ realizes $\spec$, then so does $\T'$.
\end{proof}

\begin{theorem}\label{mem not suff}
	For every $m \geq 2$, there is an LTL specification $\spec_m$ over $\{i, o_0$, $o_1$,  $\ldots, o_m\}$ such that for every $0 \leq j \leq m$, we have that $\spec_m$ is $(\emptyset$, $\{i\}$, $\{o_0, \ldots,o_{j-1}\}$, $\{o_{j},\ldots,o_m\})$-realizable by a TGE with memory~$2^{j}$ and is not $(\emptyset,\{i\}$, $\{o_0, \ldots,o_{j-1}\}$, $\{o_{j},\ldots,o_m\})$-realizable by a TGE with memory~$2^{j}-1$.
	%ofer: short
	% \begin{enumerate}
	% 	\item
	% 	$\spec_m$ is $(\emptyset,\{i\},\emptyset,\{o_1,\ldots,o_m\})$-realizable by a TGE with memory~$1$.  
	% 	\item
	% For every $1 < j \leq m$, we have that $\spec_m$ is $(\emptyset$, $\{i\}$, $\{o_1, \ldots,o_j\}$, $\{o_{j+1},\ldots,o_m\})$-realizable by a TGE with memory~$j+1$ and is not $(\emptyset,\{i\}$, $\{o_1, \ldots,o_j\}$, $\{o_{j+1},\ldots,o_m\})$-realizable by a TGE with memory~$j$.
	% \end{enumerate}
\end{theorem}
\begin{proof}
	We define $\spec_m=\psi_0 \vee \psi_1 \vee \cdots \vee \psi_m$, where for all $0 \leq j \leq m$, we have that $\psi_j=\GL(i \leftrightarrow \X^{j}o_j)$. That is, $\spec_m$ requires the existence of  $0 \leq j \leq m$ such that $o_j$ always copies $i$ with a delay of $j$ steps. 
	
	A TGE that does not guide $o_0,\ldots,o_{j-1}$, cannot realize the disjunction $\psi_0\lor\ldots\lor\psi_{j-1}$, as $i$ is hidden, and should thus guide the environment in a way that causes the realization of $\psi_{j} \vee \cdots \vee \psi_m$. This, however, must  involve at least $j$ registers that maintain the last $j$ values of $i$, enabling the TGE to realize $\psi_{j}$ and hence requires a memory of size $2^j$.
\end{proof}

\section{Solving the SGE Problem}\label{solution}

Recall that in the SGE problem, we are given an \LTL specification $\spec$ over $I \cup O$, partitions $I=\V\cup \H$ and $O=\C\cup \G$, and a bound $k\geq 1$, and we have to return a TGE that $\vhcg$-realizes $\spec$ with memory $k$. 
The main challenge in solving the SGE problem is that it adds to the difficulties of synthesis with partial visibility the fact that the system's interactions with environments that agree on the visible signals may generate different computations. Technically, the system should still behave in the same way on input sequences that agree on the visible signals. Thus, the interaction should generate the same assignments to the controlled output signals and the same sequence of programs. However, these programs may generate different computations, as they also depend on the hidden signals. 

Our solution uses {\em alternating tree automata}, and we start by defining them.

\subsection{Words, trees, and automata}
\label{subsec:words-trees-automata}
An \emph{automaton} on infinite words is $\A = \zug{ \Sigma, Q, q_0, \eta, \alpha }$, where $\Sigma$ is an alphabet, $Q$ is a finite set of states, $q_0\in Q$ is an initial state, $\eta: Q\times \Sigma \to 2^Q$ is a transition function, and $\alpha$ is an acceptance condition to be defined below. 
%For states $q,s \in Q$ and a letter $\sigma \in \Sigma$, we say that $s$ is a $\sigma$-successor of $q$ if $s \in \eta(q,\sigma)$. 
%short
%We consider automata with a total transition function. 
%That is, for every state $q\in Q$ and letter $\sigma\in \Sigma$, we have that $|\eta(q,\sigma)|\geq 1$.
If $|\eta(q, \sigma)| = 1$ for every state $q\in Q$ and letter $\sigma \in \Sigma$, then $\A$ is \emph{deterministic}.

A \emph{run}  of $\A$ on an infinite word $w = \sigma_1 \cdot \sigma_2 \cdots \in \Sigma^\omega$ is an infinite sequence of states $r = r_0\cdot r_1\cdot r_2\cdots \in Q^\omega$, such that $r_0 = q_0$, and for all $i \geq 0$, we have that $r_{i+1} \in \eta(r_i, \sigma_{i+1})$.  The acceptance condition $\alpha$ defines a subset of $Q^\omega$, indicating which runs are {\em accepting}. We consider here the \emph{B\"uchi}, \emph{co-B\"uchi}, and {\em parity\/} acceptance conditions. 
%ofer2: added back the definitions of the acceptance conditions
All conditions refer to the set  ${\it inf}(r)\subseteq Q$ of states that $r$ traverses infinitely often. Formally, ${\it inf}(r) = \{  q \in Q: q = r_i \text{ for infinitely many $i$'s}   \}$. 
In a B\"uchi automaton, the acceptance condition is $\alpha\subseteq Q$ and a run $r$ is accepting if ${\it inf}(r)\cap \alpha\neq \emptyset$. Thus, $r$ visits $\alpha$ infinitely often. Dually, in a co-B\"uchi automaton, a run $r$ is accepting if ${\it inf}(r)\cap \alpha = \emptyset$. Thus, $r$ visits $\alpha$ only finitely often. Finally, in a parity automaton $\alpha:Q\to \{1,...,k\}$ maps states to ranks, and a run $r$ is accepting if the maximal rank of a state in ${\it inf}(r)$ is even. Formally, $\max_{q \in {\it inf}(r)} \{\alpha(q)\}$ is even. 
A run that is not accepting is \emph{rejecting}.  

Note that when $\A$ is not deterministic, it has several runs on a word. 
%orna2 short
%We distinguish between two branching modes. 
If $\A$ is a {\em nondeterministic\/} automaton, then 
a word $w$ is accepted by $\A$ if there is an accepting run of $\A$ on $w$. 
If $\A$ is a {\em universal\/} automaton, then 
a word $w$ is accepted by $\A$ if all the runs of $\A$ on $w$ are accepting.
The language of $\A$, denoted $L(\A)$, is the set of words that $\A$ accepts. %Two automata are \emph{equivalent} if their languages are equivalent. 

Given a set $\Upsilon$ of directions, the {\em full $\Upsilon$-tree} is the set $T=\Upsilon^*$. 
The elements of $T$ are called {\em nodes},
and the empty word $\varepsilon$ is the {\em root} of $T$.
An {\em edge} in $T$ is a pair $\zug{x,x \cdot a}$, for $x \in \Upsilon^*$ and $a \in \Upsilon$. 
The node $x \cdot a$ is called a {\em successor} of $x$. 
A {\em path} $\pi$ of $T$ is a set $\pi \subseteq T$ such that
$\varepsilon \in \pi$ and for every $x \in \pi$, 
there exists a unique $a \in \Upsilon$ such that $x \cdot a \in \pi$.
We associate an infinite path $\pi$ with the infinite word in $\Upsilon^\omega$ obtained by concatenating the directions taken along $\pi$. 

Given an alphabet $\Sigma$, a {\em $\Sigma$-labeled $\Upsilon$-tree\/} is a pair $\zug{T,\ell}$, where $T=\Upsilon^*$ and $\ell$ labels each edge of $T$ by a letter in $\Sigma$. That is, $\ell(x,x\cdot a)\in \Sigma$ for all $x\in T$ and $a\in \Upsilon$.
Note that an infinite word in $\Sigma^\omega$ can be viewed as a $\Sigma$-labeled $\{1\}$-tree.

A {\em universal tree automaton\/} over $\Sigma$-labeled $\Upsilon$-trees
is $\A=\langle\Sigma,\Upsilon$, $Q$, $q_{0}, \eta,\alpha\rangle$, where
$\Sigma$, $Q$, $q_0$, and $\alpha$ are as in automata over words, $\Upsilon$ is the set of directions,
and
$\eta:Q \times \Sigma \times \Upsilon\rightarrow 2^Q$ is a transition function.

%orna2: some changes
Intuitively, $\A$ runs on an input $\Sigma$-labeled $\Upsilon$-tree $\zug{T,\ell}$ as follows. 
The run starts with a single copy of $\A$ in state $q_0$ that has to accept the subtree of $\zug{T,\ell}$ with root $\varepsilon$. %(that is, it has to accept $\zug{T,\ell}$).
When a copy of $\A$ in state $q$ has to accept the subtree of $\zug{T,\ell}$ with root $x$, it goes over all the directions in $\Upsilon$. For each direction $a \in \Upsilon$, it proceeds according to the transition function $\eta(q,\sigma,a)$, where $\sigma$ is the letter written on the edge $\zug{x,x \cdot a}$ of $\zug{T,\ell}$. That is, $\sigma=\ell(x,x\cdot a)$. The copy splits into $|\eta(q,\sigma,a)|$ copies: for each 
$q' \in \eta(q,\sigma,a)$, a copy in state $q'$ is created, and it has to accept the subtree with root $x \cdot a$. 

Formally, a \emph{run} of $\A$ over a $\Sigma$-labeled $\Upsilon$-tree $\zug{T,\ell}$, is a tree with directions in $\Upsilon \times Q$ and nodes labeled by pairs in $T \times Q$ that describe how the different copies of $\A$ proceed. Thus, a run is a pair $\zug{T_r,r}$, where $T_r \subseteq (\Upsilon \times Q)^*$ and $r:T_r \rightarrow T \times Q$ are defined as follows. 
\begin{itemize}
	\item
	$\varepsilon \in T_r$ and $r(\varepsilon)=\zug{\varepsilon,q_0}$. Thus, the run starts with a single copy of $\A$ that reads the root of $\zug{T,\ell}$ and is in state $q_0$.
	\item
	Consider a node $y \in T_r$ with $r(y)=\zug{x,q}$. Recall that $y$ corresponds to a copy of $\A$ that reads the node $x$ of $\zug{T,\ell}$ and is in state $q$. 
	For a direction $a \in \Upsilon$ let $\sigma_a=\ell(x, x\cdot a)$. For every state $q' \in \eta(q,\sigma_a,a)$, we have that $y \cdot \zug{a,q'} \in T_r$ and $r(y \cdot \zug{a,q'})=\zug{x \cdot a,q'}$. 
	Thus, the run sends $|\eta(q,\sigma_a,a)|$ copies to the subtree with root $x \cdot a$, one for each states in $\eta(q,\sigma_a,a)$. 
\end{itemize}

Acceptance is defined as in automata on infinite words, except that now we define, 
given a run $\zug{T_r,r}$ and an infinite path
$\pi \subseteq T_r$, the set $inf(\pi) \subseteq Q$ as the set of states that are visited along $\pi$ infinitely often, thus  $q \in inf(\pi)$ if and only if there are infinitely many nodes $y \in \pi$ for
which $r(y) \in T \times \{q\}$.
We denote by $L(\A)$ the set of all $\Sigma$-labeled $\Upsilon$-trees that $\A$ accepts.
 
%ofer2: replaced G with D in the acronyms (for deterministic)
We denote the different classes of automata by three-letter acronyms in $\{ \text{D,N,U} \} \times \{ \text{B,C,P}\} \times \{\text{W,T}\}$. The first letter stands for the branching mode of the automaton (deterministic, nondeterministic, or universal); the second for the acceptance condition type (B\"uchi, co-B\"uchi, or parity); and the third indicates we consider automata on words or trees. 
For example, UCTs are universal co-B\"uchi tree automata.

\subsection{An automata-based solution}
\label{automata sol}

Each TGE $\T=\zug{\V,H,\C,\G,S,s_0,\delta,M,m_0,\tau}$ induces a $(2^\C\times \P_{M,\H,\G})$-labeled $2^\V$-tree $\zug{(2^\V)^*,\ell}$, obtained by simulating the interaction of $\T$ with all input sequences.  
Formally, let $\delta^*:(2^\V)^* \rightarrow S$ be an extension of $\delta$ to finite sequence in $(2^\V)^*$, starting from $s_0$. Thus, $\delta^*(w)$ is the state that $\T$ visits after reading $w \in (2^\V)^*$. Formally, 
$\delta^*(\varepsilon)=s_0$, and for $w \in (2^\V)^*$ and $v \in 2^\V$, we have that $\delta^*(w \cdot v)=\delta(\delta^*(w),v)$. Then, the tree $\zug{(2^\V)^*,\ell}$ is such that $\ell(w,w\cdot v)= \tau(\delta^*(w),v)$. 
 
Given $w_V=v_1\cdot v_2\cdot v_3\cdots \in (2^\V)^\omega$, let $\ell(w_V)=\zug{c_1,p_1}\cdot \zug{c_2,p_2}\cdot \zug{c_3,p_3}\cdots\in (2^C\times\P_{M,\H,\G})^\omega$ be the sequence of labels along the path induced by $w_V$. For every $j \geq 1$, $\zug{c_j,p_j}=\ell(\zug{v_1,v_2, \ldots v_{j-1}},\zug{v_1,v_2,\ldots, v_{j}})$. 
It is convenient to think of $\ell(w_V)$ as the operation of $\T$ while interacting with an environment that generates an input sequence $(v_1\cup h_1)\cdot (v_2\cup h_2)\cdot (v_2\cup h_2)\cdots $, for some $w_\H = h_1 \cdot h_2 \cdot h_3 \cdots \in (2^\H)^\omega$. The transducer $\T$ sees only the assignments to the signals in $\V$, knowing its assignments to $\C$ and the programs for the environment, but not the assignments to $\H$, and thus neither the memory state nor the assignments to $\G$.
%
%Indeed, as $\T$ views only the assignments to the signals in $V$, it knows which assignments it makes to signals in $C$, and it knows which programs it instructs the environment to follow, but as it does not view the assignments to the signals in $H$, it does not know the state of the memory nor the assignments to the signals in $D$ that its programs induce. 

The environment, which does know $w_\H$, can complete $\ell(w_V)$ to a computation in $(2^{\V \cup \H \cup \C \cup \G})^\omega$. Indeed, given a current memory state $m\in M$ and an assignment $h\in 2^\H$ to the hidden signals, the environment can apply the last program sent $p\in \P_{M,\H,\G}$, calculate $p(m,h)=\zug{m',g}$, and thus obtain the new memory state and assignment to the guided signals. 
%
%there are unique sequences $w_D=d_1 \cdot d_2 \cdot d_3 \cdots \in (2^D)^\omega$ of assignments to the delegated signals, and $w_M=m_0,m_1,m_2, \ldots \in M^\omega$ of memories, induced by the operation of $\ell(w_V)$ on $w_H$. 
Formally, for all $j \geq 1$, we have that $\zug{m_j,g_j}=p_j(m_{j-1},h_j)$. Then, the {\em operation} of $\ell(w_V)$ on $w_\H$, denoted $\ell(w_V) \circ w_\H$, is the sequence $(v_1 \cup h_1 \cup c_1\cup g_1)\cdot (v_2 \cup h_2 \cup c_2\cup g_2) \cdots \in (2^{\V \cup \H \cup C \cup \G})^{\omega}$.

For an \LTL specification $\spec$, we say that a $(2^C\times \P_{M,\H,\G})$-labeled $2^\V$-tree $\zug{(2^\V)^*,\ell}$ is {\em $\spec$-good\/} if for all $w_V \in (2^\V)^\omega$ and $w_\H \in (2^\H)^\omega$, we have that $\ell(w_V) \circ w_\H$ satisfies $\spec$. By definition, a TGE $(\V,H,\C,\G)$-realizes $\spec$ iff its induced $(2^\C\times \P_{M,\H,\G})$-labeled $2^V$-tree is $\spec$-good.

\begin{theorem}\label{good trees}
Given an LTL specification $\spec$ over $I \cup O$, partitions $I=\V\cup \H$ and $O=\C\cup \G$, and a set $M$ of memories, we can construct a UCT $\A_\spec$ with $2^{|\spec|} \cdot |M|$ states such that $\A_\spec$ runs on $(2^\C\times \P_{M,\H,\G})$-labeled $2^\V$-trees and accepts a labeled tree $\zug{(2^\V)^*,\ell}$ iff $\zug{(2^\V)^*,\ell}$ is $\spec$-good.
\end{theorem}

\begin{proof}
Given $\spec$, let $U_\spec=\zug{2^{I\cup O},Q,q_0,\eta,\alpha}$ be a UCW over the alphabet $2^{I \cup O}$ that recognizes $L_\spec$. We can construct $U_\spec$ by dualizing an NBW for $\neg \spec$. By \cite{VW94}, the latter is of size exponential in $|\spec|$, and thus, so is $U_\spec$. 

We define $\A_\spec=\zug{2^\C\times \P_{M,H,\G},2^\V,Q \times M,\zug{q_0,m_0},\eta',\alpha \times M}$, where $m_0\in M$ is arbitrary, and $\eta'$ is defined for every $\zug{q,m} \in Q\times M$, $v \in 2^\V$, and $\zug{c,p} \in  2^\C\times \P_{M,\H,\G}$, as follows. 
$$
\eta'(\zug{q,m},\zug{c,p},v)=$$
$$\bigcup_{h \in 2^\H} \eta(q,v\cup h\cup c\cup p_\G(m,h)) \times \{p_M(m,h)\}.$$

Thus, when the label to direction $v \in 2^\V$ is $\zug{c,p}$, the UCT $\A_\spec$ sends copies that correspond to all possible assignments in $2^\H$, where for every $h \in 2^\H$, a copy is sent for each state in $\eta(q,v\cup h\cup c\cup p_\G(m,h))$, all with the same memory $p_M(m,h)$. Accordingly, for all $2^C\times\P_{M,\H,\G}$-labeled $2^\V$-tree $\zug{(2^\V)^*,\ell}$, there is a one to one correspondence between every infinite branch in the run tree $\zug{T_r,r}$ of $\A$ over $\zug{(2^\V)^*,\ell}$, and an infinite run of the UCW $\A_\spec$ over $\ell(w_V)\circ w_\H$, for some $w_V\in (2^\V)^\omega$ and $w_\H\in (2^\H)^\omega$. It follows that $\A$ accepts a tree $\zug{(2^\V)^*,\ell}$ iff $\zug{(2^\V)^*,\ell}$ is $\spec$-good.
\end{proof}

The construction of $\A_\spec$ allows us to conclude the following complexity result of the SGE problem. 

\begin{theorem}\label{bounded 2exp}
The LTL SGE problem is 2EXPTIME-complete. Given an LTL formula $\spec$ over sets of signals $\V$, $H$, $C$, and $\G$, and a integer $k$, deciding whether $\spec$ is $(\V,H,C,\G)$-realizable by a TGE with memory $k$ can be done in time doubly exponential in $|\spec|$ and exponential in $k$. 
\end{theorem}

\begin{proof}
We start with the lower bound which follows immediately by the fact that traditional LTL synthesis is 2EXPTIME-complete~\cite{Ros92}. The traditional synthesis of a specification $\spec$ above $I\cup O$ is cleary reduced to the synthesis of a TGE that $(I,\emptyset,O,\emptyset)$-realizes $\spec$. Indeed, a TGE with $\H=\emptyset$ and $\G=\emptyset$ is simply a traditional transducer as its transmitted programs are redundant and do not affect the generated computation. Thus, since LTL synthesis is 2EXPTIME-hard, we conclude that so is the SGE problem.

We continue and prove the upper bound. Let $M$ be a memory of size $k$, and let $\A$ be the UCT constructed for $\spec, \V,H,\C,G$, and $M$ in \autoref{good trees}. As $\A$ accepts exactly all  $\spec$-good $(2^\C\times \P_{M,\H,\G})$-labeled $2^\V$-trees, we can reduce
SGE to the non-emptiness problem for $\A$, returning a witness when it is non-empty. 

In \cite{KV05c} the authors solved the non-emptiness problem for UCTs by translating a UCT $\A$ into an NBT $\A'$ such that $L(\A')\subseteq L(\A)$ and $L(\A')\neq \emptyset$ iff $L(\A)\neq \emptyset$. The translation goes from a UCT with $n$ states to an NBT with $2^{O(n^2\log n)}$ states.\footnote{The construction in \cite{KV05c} refers to automata with labels on the nodes rather than the edges, but it can be easily adjusted to edge-labeled trees.}  

Consider an NBT with state space $S$ that runs on $\Sigma$-labeled $\Upsilon$-trees. Each nondeterministic transition of the NBT maps an assignment $\Sigma^\Upsilon$ (of labels along the branches that leave the current node) to a set of functions in $S^\Upsilon$ (describing possible labels by states of the successors of the current node in the run tree). Accordingly, non-emptiness of the NBT can be solved by deciding a B\"uchi game in which the OR-vertices are the states in $S$, and the AND-vertices are functions in $S^\Upsilon$ \cite{GH82}.

In our case, as $\A$ has $k\cdot 2^{|\spec|}$ states, we have that $S$ is of size $2^{O(k^2\log k2^{O(|\spec|)})}$ and 
$|\Upsilon|=2^{|\V|}$. Hence, $|S^\Upsilon|$, which is the main factor in the size of the game, is $2^{O(k^2\log k2^{O(|\spec|)} 2^{|\V|})}$.
Since $|\V|\leq |\spec|$ and since B\"uchi games can be solved in quadratic time \cite{VW86a}, we end up with the required complexity. 
\end{proof}

\subsection{Bounding the size of \texorpdfstring{$M$}{M}}
\label{Bound M}

As discussed in Example~\ref{illustrate memory}, programs that can only refer to the current assignment of the hidden signals may be too weak: in some specifications, the assignment to the guided output signals have to depend on the history of the interaction so far. The full history of the interaction is a finite word in $(2^{I \cup O})^*$, 
and so apriori, an unbounded memory is needed in order to remember all possible histories. 
A deterministic automaton $\D_\spec$ for $\spec$ partitions the infinitely many histories in $(2^{I \cup O})^*$ into finitely many equivalence classes. Two histories $w,w' \in (2^{I \cup O})^*$ are in the same equivalence class if they reach the same state of $\D_\spec$, which implies that $w \cdot t \models \spec$ iff  $w' \cdot t \models \spec$ for all $t \in (2^{I \cup O})^\omega$. In the case of traditional synthesis, we know that a transducer that realizes $\spec$ does not need more states than $\D_\spec$. Intuitively, if two histories of the interaction are in the same equivalence class, the transducer can behave the same way after processing them. 

In the following theorem we prove that the same holds for the memory used by a TGE: if two histories of the interaction reach the same state of $\D_\spec$, there is no reason for them to reach different memories. Formally, we have the following.

\begin{theorem}\label{TGE automata mem}
Consider a specification $\spec$, and let $Q$ be the set of states of a DPW for $\spec$.
If there is a TGE that $\vhcg$-realizes $\spec$ with memory $M$, then there is also a TGE that $\vhcg$-realizes $\spec$ with memory $Q$.
\end{theorem}
\begin{proof}
	Let $\T=\zug{\V,H,\C,\G,S,s_0,\delta,M,m_0,\tau}$ be a TGE that $\vhcg$-realizes $\spec$, and let $\T_{(\V/\C)}=\zug{\V,\C,S,s_0,\delta,\tau_C}$ be the $(\V,C)$-transducer induced by $\T$. The labeling function $\tau_C:S\times 2^\V\rightarrow 2^C$ is obtained by projecting $\tau$ on its $2^C$ component. Let $\D_\spec=\zug{2^{I\cup O},Q,q_0,\eta,\alpha}$ be a DPW for $\spec$. We show that there exists a labeling function $\tau_{\G}:S\times 2^\V\rightarrow \P_{Q,\H,\G}$ such that $\T_Q=\zug{V,\H,\C,\G,S,s_0,\delta,Q,q_0,\tau'}$, with $\tau'(s,v)=\zug{\tau_C(s,v),\tau_{\G}(s,v)}$ is a TGE that $\vhcg$-realizes $\spec$. 
	
	Consider the following two player game ${\cal G}_{\T,\D_\spec}$ played on top of the product of $\T_{(\V/C)}$ with $\D_\spec$.
	The game is played between the environment, which proceeds in positions in $S\times Q$, and the system, which proceeds in positions in $S\times Q\times 2^I$. The game is played from the initial position $\zug{s_0,q_0}$, which belongs to the environment. From a position $\zug{s,q}$, the environment chooses an assignment $i\in 2^I$ and moves to $\zug{s,q,i}$. From a position $\zug{s,q,i}$, the system chooses an assignment $g\in 2^\G$ and moves to $\zug{s',q'}$,  where $s'=\delta(s,i\cap \V)$ and $q'=\eta(q,i\cup \tau_C(s,i\cap \V)\cup g)$. The system wins the game if the $Q$ component of the generated play satisfies $\alpha$. 
	
	Note that the system wins iff there exists a strategy $f:(2^I)^*\rightarrow 2^\G$ that generates for each $w_I\in (2^I)^\omega$, a word $f(w_I)\in (2^\G)^\omega$ such that $w_I|_H\oplus \T_{(\V/C)}(w_I|_V)\oplus f(w_I)\in L(\D_\spec)$. Indeed, the $Q$-component in the outcome of the game when the environment plays $w_I\in (2^I)^\omega$ is precisely the run of $\D_\spec$ on the word $w=w_I\cup w_C\cup f(w_I)$, where $w_C$ is the word generated by the $(\V/C)$-transducer $\T_{(\V/C)}$ on the input word $w_I|_V$, and, by definition $\T_{(\V/C)}(w_I|_V)=w_I|_V\oplus w_C$. Moreover, by the memoryless determinacy of parity games, such a strategy $f$ exists iff there exists a winning positional strategy $g:S\times Q\times 2^I\rightarrow 2^\G$. Thus, instead of considering the entire history in $(2^I)^*$, the system has a winning strategy that only depends on the current position in $S\times Q$. 
	
	For all $s\in S$ and $v\in 2^\V$, let $p_{s,v}:Q\times 2^\H\rightarrow Q\times 2^\G$ be defined by $p_{s,v}(q,h)=\zug{q',d}$, where $d=g(s,q,v\cup h)$ and $q'=\eta(q,v\cup h\cup \tau_C(s,v)\cup d)$. Finally, let $\tau_{Q}:S\times 2^\V\rightarrow \P_{Q,\H,\G}$ be defined by $\tau_{Q}(s,v)=p_{s,v}$, and $\T_Q=\zug{\V,\H,\C,\G,S,s_0,\delta,Q,q_0,\tau_Q}$ be the TGE with memory $Q$ that is obtained from $\T$ by replacing $M$, $m_0$,  and $\tau$ with $Q$, $q_0$, and $\tau_Q$, respectively.
	
	It is not hard to see that for all $w_I\in (2^I)^\omega$, the computation $\T_Q(w_I)$ is exactly the outcome of the game when Sys plays with the winning strategy $g$. I.e., $\T_Q(w_I)=w_I|_\H\oplus \T_{(\V/C)}(w_I|_V)\oplus g(w_I)$. Hence, $\T_Q(w_I)\in L(\D_\spec)$, for all $w_I\in (2^I)^\omega$, as required. 
\end{proof}
 
\begin{theorem}\label{TGE automata mem 2}
If there is a TGE that $\vhcg$-realizes an LTL specification $\spec$, then there is also a TGE that $\vhcg$-realizes $\spec$ with memory doubly exponential in $|\varphi|$, and this bound is tight. 
\end{theorem}

\begin{proof}
The upper bound follows from \autoref{TGE automata mem} and the doubly-exponential translation of LTL formulas to DPWs \cite{VW94,Saf88,EKS20}. The lower bound follows from the known doubly-exponential lower bound on the size of transducers for LTL formulas \cite{Ros92}, applied when $I=\H$ and $O=\G$.
\end{proof}

While \autoref{TGE automata mem 2} is of theoretical interest, we find the current formulation of the SGE problem, which includes a bound on $M$, more appealing in practice: recall that SGE is doubly-exponential in $|\spec|$ and exponential in $|M|$. As $|\D_\spec|$ is already doubly exponential in $|\spec|$, solving SGE with no bound on $M$ results in an algorithm that is triply-exponential in $|\spec|$. Thus, it makes sense to let the user provide a bound on the memory.

\section{Programs}
\label{sec programs}
Recall that TGEs generate in each transition a program $p:M \times 2^\H \rightarrow M \times 2^\G$ that instructs the environment how to update its memory and assign values to the guided output signals. In this section we discuss ways to represent programs efficiently, and, in the context of synthesis, restrict the set of programs that a TGE may suggest to its environment without affecting the outcome of the synthesis procedure. Note that the number of programs in $\P_{M,\H,\G}$ is $2^{(\log |M|+|\G|) \cdot |M| \cdot 2^{|\H|}}$. Our main goal is to reduce the domain $2^{\H}$, which is the most dominant factor. 

Naturally, the reduction depends on the specification we wish to synthesize. For an LTL formula $\psi$, let $\prop{\psi}$ be the set of maximal predicates over $I \cup O$ that are subformulas of $\psi$. Formally, $\prop{\psi}$ is defined by an induction on the structure of $\psi$ as follows. 
\begin{itemize}
	\item
	If $\psi$ is a propositional assertions, then $\prop{\psi}=\{\psi\}$.
	\item
	Otherwise, $\psi$ is of the form $* \psi_1$ or $\psi_1 * \psi_2$ for some (possibly temporal) operator $*$, and $\prop{\psi}=\prop{\psi_1} \cup \prop{\psi_2}$.
\end{itemize}

Note that the definition is sensitive to syntax. For example, the formulas $(i_1 \vee o)\wedge \X  i_2$ and $(i_1 \wedge \X  i_2) \vee (o \wedge \X  i_2)$ are equivalent, but have different sets of maximal propositional assertions. Indeed, 
$\prop{(i_1 \vee o)\wedge \X  i_2}=\{i_1\vee o, i_2\}$, whereas $\prop{(i_1 \wedge \X  i_2) \vee (o \wedge \X  i_2)}=\{i_1,i_2,o\}$. 

It is well known that the satisfaction of an LTL formula $\psi$ in a computation $\pi \in (2^{I \cup O})^\omega$ depends only on the satisfaction of the formulas in $\prop{\psi}$ along $\pi$: if two computations agree on $\prop{\psi}$, then they also agree on $\psi$. Formally, for two assignments $\sigma,\sigma' \in 2^{I \cup O}$, and a set $\Theta$ of predicates over $I \cup O$, we say that $\sigma$ and $\sigma'$ {\em agree on\/} $\Theta$, denoted $\sigma \approx_{\Theta} \sigma'$, if for all $\theta \in \Theta$, we have that $\sigma \models \theta$ iff $\sigma' \models \theta$. Then, two computations $\pi=\sigma_1\cdot \sigma_2\cdots$ and $\pi' =\sigma'_1\cdot\sigma'_2\cdots$ in $(2^{I \cup O})^\omega$ agree on $\Theta$, denoted $\pi \approx_{\Theta} \pi'$, if for all $j \geq 1$, we have $\sigma_j \approx_{\Theta} \sigma'_j$. 

\begin{proposition}
	\label{agree on prop}
	Consider two computations $\pi,\pi' \in (2^{I \cup O})^\omega$. For every LTL formula $\psi$, if $\pi \approx_{\prop{\psi}} \pi'$, then $\pi \models \psi$ iff $\pi' \models \psi$. 
\end{proposition}

As we shall formalize below, Proposition~\ref{agree on prop} enables us to restrict the set of programs so that only one computation from each equivalence class of the relation $\approx_{\prop{\psi}}$ may be generated by the interaction of the system and the environment. 

For an LTL formula $\psi$, let $\cl{\H}(\psi)$ be the set of maximal subformulas of formulas in $\prop{\psi}$ that are defined only over signals in $\H$. 
Formally, $\cl{\H}(\psi)=\bigcup_{\theta\in \prop{\psi}}\cl{\H}(\theta)$, where $\cl{\H}(\theta)$ is defined for a propositional formula $\theta$ as follows.
\begin{itemize}
	\item
	If $\theta$ is only over signals in $\H$, then $\cl{\H}(\theta)=\{\theta\}$.
	\item
	If $\theta$ is only over signals in $\V \cup O$, then $\cl{\H}(\theta)=\emptyset$.
	\item
	Otherwise, $\theta$ is of the form $\neg \theta_1$ or $\theta_1 \ast \theta_2$ for $\ast\in \set{\vee,\wedge}$, in which case $\cl{\H}(\theta)=\cl{\H}(\theta_1)$ or $\cl{\H}(\theta)=\cl{\H}(\theta_1) \cup \cl{\H}(\theta_2)$, respectively.
\end{itemize}

For example, if $\H=I=\{i_1,i_2,i_3\}$ and $O=\{o\}$, then $\cl{\H}((i_1 \vee i_2)\wedge (i_3\vee o))=\{i_1 \vee i_2,i_3\}$. 
Note that formulas in $\cl{\H}(\psi)$ are over $\H$, and that the relation $\approx_{\cl{\H}(\psi)}$ is an equivalence relation on $2^\H$. 
We say that a program $p: M \times 2^\H \rightarrow M \times 2^\G$ is {\em tight\/} for $\psi$ if for every memory $m \in M$ and two assignments $h,h' \in 2^{\H}$, if $h \approx_{\cl{\H}(\psi)} h'$, then $p(m,h)=p(m,h')$. %Thus, assignments to $H$ that agree on all the formulas in $\cl{H}(\psi)$ are mapped by $p$ to the same assignments to $D$. For $h \in 2^H$, let $[h]$ denote the equivalence class of $h$ in the $\approx_{\cl{H}(\psi)}$. 

In Theorem~\ref{only memory tight} below, we argue that in the context of LTL synthesis, one can always restrict attention to tight programs (see Example~\ref{app ex3} for an example for such a restriction). 

\begin{theorem}
	\label{only memory tight}
	If $\spec$ is $\vhcg$-realizable by a TGE, then it is $\vhcg$-realizable by a TGE that uses only programs tight for $\spec$. 
\end{theorem}

\begin{proof}
Let $\T=\zug{\V,H,C,\G,S,s_0,\delta,M,m_0,\tau}$ be a TGE that realizes $\spec$ with arbitrary programs. We define a labeling function $\tau':S \times 2^\V \rightarrow 2^C \times \P_{M,H,\G}$ that uses only programs tight for $\spec$, and argue that the TGE $\T'=\langle \V, H, C, \G$, $S, s_0$, $\delta, M$, $m_0, \tau' \rangle$, obtained from $\T$ by replacing $\tau$ by $\tau'$,  realizes $\spec$.

Recall that the relation $\approx_{\cl{\H}(\spec)}$ is an equivalence relation on $2^\H$. Let ${\it rep}:2^\H \rightarrow 2^\H$ map each assignment $h \in 2^\H$ to an assignment that represents the $\approx_{\cl{\H}(\spec)}$-equivalence class of $h$; for example, we can define the representative assignment to be the minimal equivalent assignment according to some order on $2^\H$. 

We define $\tau'$ such that for every $s \in S$ and $v \in 2^\V$ with $\tau(s,v)=\zug{c,p}$, we have $\tau'(s,v)=\zug{c,p'}$, where the program $p'$ is such that for all $m\in M$ and $h \in 2^\H$, we have that $p'(m,h)=p(m,{\it rep}(h))$. Clearly, $\T'$ uses only tight programs. Indeed, all the assignments in the same equivalence class of $\approx_{\cl{\H}(\spec)}$ are mapped to the same program. 

We continue and prove that $\T'$ realizes $\spec$. Consider an input sequence $\wi=\zug{v_1,h_1}\cdot \zug{v_2,h_2}\cdot \zug{v_3,h_3}\cdots \in (2^\V \times 2^\H)^\omega$. Let $r=s_0\cdot s_1\cdot s_2\cdots \in S^\omega$ be the run of $\T$ on $w_I$, and let $\zug{c_1,p_1}\cdot \zug{c_2,p_2}\cdot \zug{c_3,p_3}\cdots \in (2^C \times \P_{M,\H,\G})^\omega$ be the sequence of labels along the transitions of $r$. Thus, $\zug{c_j,p_j}=\tau(s_{j-1},v_j)$ for all $j\geq 1$. Since $\T'$ differs from $\T$ only in the programs that $\tau'$ generates, the run $r$ is also the run of $\T'$ on $w_I$, and the sequence of labels along the transitions in it is $\zug{c_1,p'_1}\cdot \zug{c_2,p'_2}\cdot \zug{c_3,p'_3}\cdots$, where $\zug{c_j,p'_j}=\tau'(s_{j-1},v_j)$ for all $j \geq 1$. 
%orna1: I think that the new sequence of registers should be defined here
Hence, if we let $\zug{m_j,d_j}=p'_j(m_{j-1},h_j)$ for all $j\geq 1$, then $w_G=d_1\cdot d_2\cdot d_3\cdots\in (2^\G)^\omega$ is sequence of assignments to the output signals in $\G$ that $\T'$ instructs the environment to perform when it reads $w_I$.

Consider now the input sequence $w'_I=\zug{v_1,{\it rep}(h_1)}$, $\zug{v_2,{\it rep}(h_2)}$, $\zug{v_3,{\it rep}(h_3)}$, $\ldots \in (2^\V \times 2^\H)^\omega$. Since $w_I$ and $w'_I$ agree on the input signals in $\V$, and since $\delta$ and $\tau$ depend only on the input signals in $V$, the run of $\T$ on $w'_I$ is also $r$, and the sequence of labels along $r$ is also $\zug{c_1,p_1},\zug{c_2,p_2},\zug{c_3,p_3},\ldots\in (2^C\times \P_{\H,\G})^\omega$. Thus if we define $m'_0=m_0$ and $\zug{m'_j,d'_j}=p_j(m'_{j-1},{\it rep}(h_j))$, for all $j\geq 1$, then $w'_G=d'_1\cdot d'_2\cdot d'_3\cdots\in (2^\G)^\omega$ is the sequence of assignments to the output signals in $\G$ that $\T$ instructs the environment to perform when it reads $w'_I$.

We prove that $\zug{m_j,d_j}=\zug{m'_j,d'_j}$ for all $j\geq 1$. Recall that by definition of $\tau'$, for all $m\in M$ and $h \in 2^\H$, we have that $p'_j(m,h)=p_j(m,{\it rep}(h))$. Hence, for all $j\geq 1$ we have $\zug{m'_j,d'_j}=p_j(m'_{j-1},{\it rep}(h_j))=p'_j(m'_{j-1},h_j)$. Thus, as $m'_0=m_0$, it follows by induction that $\zug{m_j,d_j}=p'_j(m'_{j-1},h_j)=p'_j(m_{j-1},h_j)=\zug{m'_j,d'_j}$, for all $j\geq 1$. In particular, $w'_G=w_G$, and so $\T$ generates on $w'_I$ the same assignments to the signals in $\G$ as $\T'$ generates on $w_I$. 

Thus, for the input sequence $w'_I$, the TGE $\T$ generates the computation $\T(w'_I)=\sigma_1 \cdot \sigma_2 \cdot \sigma_3 \cdots \in (2^{I\cup O})^\omega$, with $\sigma_j=v_j \cup {\it rep}(h_j) \cup c_j \cup d_j$, for all $j \geq 1$. Since $\T$ realizes $\spec$, we know that $\pi \models \spec$. Then, for the input sequence $w_I$, the TGE $\T'$ generates the computation $\T'(w_I)=\sigma'_1 \cdot \sigma'_2 \cdot \sigma'_3 \cdots\in (2^{I\cup O})^\omega$, with $\sigma'_j=v_j \cup h_j \cup c_j \cup d_j$, for all $j \geq 1$. 

Note that the computations $\T(w'_I)$ and $\T'(w_I)$ agree on $\V \cup O$, and also agree on all the formulas in $cl_\H(\spec)$. Hence, as all the formulas in $\prop{\spec}$ are composed of predicates over $\V \cup O \cup cl_\H(\spec)$, we have that 
$\T(w'_I) \approx_{\prop{\spec}} \T'(w_I)$. By Proposition~\ref{agree on prop}, we thus have that $\T'(w_I) \models \spec$, and we are done.
\end{proof}

\begin{example}{\bf [Simplifying programs]}
\label{app ex3}
	Let $\H=\{i_1,i_2\}$ and $\G=\{o\}$. Consider the propositional specification $\spec=(i_1 \wedge i_2) \leftrightarrow \neg o$. Note that there are $16$ programs in $\P_{\{i_1,i_2\},\{o\}}$ (where no memory is used). For example, the program $p:2^\H \rightarrow \{T,F\}$ assigns values to $o$ as follows.  
	$$
	\begin{array}{|c|c||c|}
		i_1 & i_2 & p(i_1,i_2)  \\ \hline \hline
		F & F & T  \\ \hline
		F & T & T  \\ \hline
		T & F & F  \\ \hline
		T & T & T  \\ \hline
	\end{array}
	$$
	Let $\theta=i_1 \wedge i_2$. Note that $\cl{\H}(\spec)= \{\theta\}$. Hence, $p$ above is not tight, as $\{i_1\} \approx_{\{\theta\}} \{i_2\}$, yet $p(\{i_1\}) \neq p(\{i_2\})$. Also, among the $16$ programs in $\P_{\{i_1,i_2\},\{o\}}$, only four are tight:
	$$
	\begin{array}{|c||c|c|c|c|}
		\theta & p_1(\theta) &  p_2(\theta) & p_3(\theta) & p_4(\theta)  \\ \hline \hline
		F & F & F & T & T   \\ \hline
		T & F & T & F & T   \\ \hline
	\end{array}
	$$
	Among these tight programs, program $p_3$ satisfies the specification, in the sense that $\spec$ is valid for all assignments in $2^I$. Indeed, program $p_3$ induces the following assignments:
	$$
	\begin{array}{|c|c|c|c|c|}
		i_1 & i_2 & \theta & o & \spec \\ \hline \hline
		F & F & F & T & T \\ \hline
		F & T & F & T & T \\ \hline
		T & F & F & T & T \\ \hline
		T & T & T & F & T \\ \hline
	\end{array}
	$$
	Thus, $\spec$ can be realized by a TGE that  uses a tight program. 
\hfill \qed \end{example}

Theorem~\ref{only memory tight} enables us to replace the domain $M \times 2^\H$ by the more restrictive domain $M \times 2^{\cl{\H}(\spec)}$. Below we discuss how to reduce the domain further, taking into account the configurations in which the programs are going to be executed.

Recall that whenever a TGE instructs the environment which program to follow, the signals in $\V$ and $\C$ already have an assignment. Indeed, $\tau: S \times 2^\V \rightarrow 2^\C \times \P_{M,\H,\G}$, and when $\zug{c,p} \in \delta(s,v)$, we know that the program $p$ is going to be executed when the signals in $V$ and $C$ are assigned $v$ and $c$. 

Below we discuss and demonstrate how the known assignments to $\V$ and $\C$ can be used to make the equivalence classes of $2^\H$ coarser.
Note that the formulas in $\cl{\H}(\theta)$ are over $\H$, thus the assignment to the signals in $\V \cup \C$ does not affect them. Rather, it can ``eliminate" from  $\prop{\spec}$ formulas that are evaluated to $\True$ or $\False$. 
For example, if $\prop{\spec}=\{(v_1 \wedge h_1 \wedge h_3) \vee d_1, (c_1 \vee h_2) \wedge d_2\}$, then an assignment of $\True$ to $v_1$ and $c_1$ simplifies the formulas to $\{(h_1 \wedge h_3) \vee d_1, d_2\}$. 
Formally, we have the following.

For a set $\Theta$ of propositional formulas over $I \cup O$ and an assignment $f \in 2^{\V \cup \C}$, let $\Theta_{|f}$ be the set of non-trivial (that is, different from $\True$ or $\False$) propositional formulas over $\H \cup \G$ obtained by simplifying the formulas in $\Theta$ according to the assignment $f$ to  the signals in $\V\cup \C$. For example, if
$\Theta=\{(v_1 \wedge h_1) \vee d_1, (c_1 \vee h_2 ) \wedge d_2\}$, and $f(v_1)=f(c_1)=\False$, then $\Theta_{|f}=\{d_1, h_2 \wedge d_2\}$.

For an LTL formula $\psi$, and an assignment $f \in 2^{\V \cup \C}$, let $\cl{\H,f}(\psi)$ be the set of maximal subformulas of formulas in $\prop{\psi}_{|f}$ % that are 
defined only over signals in $\H$.
Formally, $\cl{\H,f}(\psi)=\bigcup_{\theta\in \prop{\psi}_{|f}}\cl{\H}(\theta)$.
For example, for $\spec$ with $\prop{\spec}=\{(v_1 \wedge h_1 \wedge h_3) \vee d_1, (c_1 \vee h_2) \wedge d_2\}$,  while $\cl{\H}{(\spec)}=\{h_1 \wedge h_3,h_2\}$, we have the simplification presented in the table below. 
$$
\begin{array}{|c|c||c|}
	f(v_1) & f(c_1) & \cl{\H,f}(\spec)  \\ \hline \hline
	F & F & \{h_2\}  \\ \hline
	F & T & \emptyset  \\ \hline
	T & F & \{h_1 \wedge h_3,h_2\}  \\ \hline
	T & T & \{h_1 \wedge h_3\}  \\ \hline
\end{array}
$$

Now, we say that a program $p: M \times 2^\H \rightarrow M \times 2^\G$ is {\em $f$-tight\/} if for every memory $m \in M$ and two assignments $h,h' \in 2^{\H}$, if $h \approx_{\cl{\H,f}(\spec)} h'$, then $p(m,h)=p(m,h')$.

Then, a TGE is tight if for every state $s \in S$ and assignments $v \in \V$ with $\delta(s,v)=\zug{c,p}$, we have that $p$ is $(v \cup c)$-tight. It is easy to extend the proof of Theorem~\ref{only memory tight} to tight TGEs. Indeed, now, for every position $j \geq 1$ in the computation, the equivalence class of $h_j \in 2^\H$ is defined with respect to $\cl{\H,v_j \cup c_j}(\spec)$, and all the considerations stay valid.

\begin{remark}{\rm 
Typically, a specification to a synthesized system is a conjunction of sub-specifications, each referring to a different functionality of the system. Consequently, the assignment to each output signal may depend only on a small subset of the input signals -- these that participate in the sub-specifications of the output signal. For example, in the specification $\spec=\GL(o_1\iff i_1)\wedge \GL(o_2\iff i_2)$, the two conjuncts are independent of each other, and so the assignment to $o_1$ can depend only on $i_1$, and similarly for $i_1$ and $o_2$. Accordingly, further reductions to the set of programs can be achieved by decomposing programs to sub-programs in which different subsets of $\H$ are considered when assigning values to different subsets of $\G$. In addition, by analyzing dependencies within each sub-specification, the partition of the output signals and their corresponding ``affecting sets of hidden signals'' can be refined further. As showed in \cite{ABCF24}, finding dependent output signals leads to improvements in state-of-the-art synthesis algorithms. \hfill \qed}
\end{remark}

\section{Viewing a TGE as a Distributed System}
Consider a TGE $\T$ with memory $M$ that $\vhcg$-realizes a specification. 
%orna2: repetition, so I remove
%In each round of its interaction with the environment, $\T$ receives an assignment to the signals in $\V$, and generates an assignment to the signals in $\C$ and a program that is sent to the environment. This program directs the environment on updating its memory and generating an assignment to the signals in $G$, based on the assignment to the signals in $\H$ and the current memory state. 
As shown in Figure~\ref{arch dist} (left), the TGE's operation can be viewed as two distributed processes executed together: the TGE $\T$ itself, and a transducer $\T_G$ with state space $M$, implementing $\T$'s instructions to the environment. In each round, the transducer $\T_G$ receives from the environment an assignment to the signals in $\H$, received from $\T$ a program, and uses both in order to generate an assignment to the signals in $G$ and update its state.

\begin{figure}[!htb] 
\centering 
\includegraphics[width=0.15\textwidth]{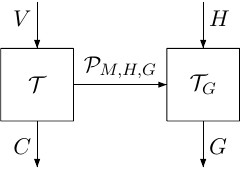}  
\hspace{1.8cm}
\includegraphics[width=0.15\textwidth]{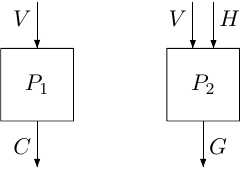} 
 \caption{A TGE that interacts with a guided environment (left) and the corresponding distributed system architecture (right).} 
 \label{arch dist} 
 \end{figure}

We examine whether viewing TGEs this way can help reduce SGE to known algorithms for the synthesis of distributed systems. 
%orna2 some rewrites
We argue that the approach here, where we do not view $\T_G$ as a process in a distributed system, is preferable. In distributed-systems synthesis, we are given a specification $\spec$ and an \emph{architecture} describing the communication channels among processes. The goal is to design strategies for these processes so that their joint behavior satisfies $\spec$. Synthesis of distributed systems is generally undecidable \cite{PR90}, primarily due to \emph{information forks} -- processes with incomparable information (e.g., when the environment sends assignments of disjoint sets of signals to two processes) \cite{Sch08}. The SGE setting corresponds to the architecture in Figure~\ref{arch dist}: Process $P_1$ is the TGE $\T$, which gets assignments to $\V$ and generates assignments to $\C$. Instead of designing $P_1$ to generate instructions for the environment, the synthesis algorithm also returns $P_2$, which instructs the environment on generating assignments to $\G$. The process $P_2$ gets (in fact generates) assignments to both $\V$ and $\H$, eliminating information forks, making SGE solvable by solving distributed-system synthesis for this architecture. A solution in the TGE setting, that is composed of a TGE $\T$ and an environment transducer $\T_G$, induces a solution in the distributed setting: $P_1$ follows $\T$, and $P_2$ simulates the joint operation of $\T$ and $\T_G$, assigning values to $\G$ as instructed by $\T$. Conversely, a TGE can encode through $\P_{M,\H,\G}$ the current assignment to $V$ together with a description of the structure of $P_2$, achieving the architecture in Figure~\ref{arch dist} (right).

Using programs in $\P_{M,\H,\G}$ goes beyond sending $\V$'s values, which are already known to the environment. Programs leverage the TGE's computation, particularly its current state, to save resources and utilize less memory. Not using the communication channel between the TGE and the environment could result in a significant increase in the size of process $P_2$. For example, when $|M|=1$ (specifications where $\G$'s assignment depends only on $\V$'s history and $\H$'s current assignment), the process $P_2$ is redundant. An explicit example is in Theorem~\ref{aux quad}, similar to the proof of \autoref{tce comp quad}. As demonstrated in Section~\ref{sec programs}, programs in $\P_{M,\H,\G}$ can be described symbolically.
Formally, we have the following.

\begin{theorem}\label{aux quad}
	For every $k\geq 1$, there exists a specification $\spec_k$ over $\V\cup \H\cup \G$, such that $\spec_k$ is $\vhcg$-realizable by a TGE with a set of states and a memory set both of size $O(k)$, yet the size of $P_2$ in a distributed system that realizes $\spec_k$ is at least $\Omega(k^2)$.
\end{theorem}

\begin{proof}
Let $\V=\set{v}, \H=\set{h}$, and $\G=\set{d}$. Consider the specification $\spec_k= Fd \leftrightarrow ((F^k v)\wedge (F^k h))$, where the operator $F^k$ stands for ``at least $k$ occurrences in the future'' (see formal definition in \autoref{tce comp quad}). Accordingly, the specification $\spec_k$ requires that $d$ is eventually turned on iff both signals $v$ and $h$ are turned on at least $k$ times. 

We prove that in the distributed setting, Process $P_2$ needs to implement two $k$-counters simultaneously, while a TGE may decompose the two counters between the system and the environment, implying the stated quadratic saving. 

Note that the synthesis of $\spec_k$ in a distributed setting forces the process $P_2$ to implement two $k$-counters, one for the number of times $v$ is on, and one for the number of times $h$ is on. Accordingly, $P_2$ needs at $\Omega(k^2)$ states. On the other hand, synthesis of $\spec_k$ by a TGE enables a decomposition of the two counters: The TGE  $\T$ maintains a counter for $v$, and the environment transducer $\T_G$ needs to implement only a counter for $h$. Indeed, as long that the $v$ counter has a value below $k$, the TGE $\T$ sends to $\T_G$ a program that instruct it to assign $d$ with $\False$, and update its state according to $h$.  Once the counter of $v$ reaches $k$, the TGE instructs $\T_G$ to turn $d$ on if its counter of $h$ reached $k$. Thus, there is a TGE with state space and memory set both linear in $k$.
\end{proof}

\section{SGE for Branching-Time Specifications}

Branching-time specifications allow reasoning about the behavior of systems across multiple possible futures. 
Formulas in the temporal logic \ctls \cite{Lam80,EH86} include, in addition to the Boolean and temporal operators of LTL, the path quantifiers $A$ (``for all paths") and $E$ ("there exists a path"), and its semantics is defined with respect to computation trees. 
In the context of synthesis, different futures correspond to different possible behaviors of the environment. For example, the \ctls formula $AGEF{\it release}$ specifies that for all input sequences, the environment may always deviate to an input sequence that would eventually cause a release. Note that such a behavior cannot be specified in LTL 
%ofer3: added the following to emphasize that is not a specific problem with LTL, but rather an inherent deference between only specifying "good-computations" as opposed to specifying "good systems" in a way that allows relating its different computations.
(or to any formalism that only uses universal path quantification).   
The ability to refer to different futures is particularly relevant to SGE, where the interaction between the system and the environment can lead to diverse outcomes depending on the environment's hidden information. In this section, we extend the SGE framework to handle branching-time specifications, enabling the synthesis of systems that satisfy such specifications in the presence of guided environments.

%orna3: some changes
%ofer3: formalizes-->formalisms? yes
Traditionally, the semantics of \ctls is defined with respect to computation trees induced by Kripke structures, namely trees with node (rather than edge) labeling. We first define such trees. For a sets $\Upsilon$ of directions and an alphabet $\Sigma$, a \emph{$\Sigma$-node-labeled $\Upsilon$-tree\/} is a pair $\zug{T,\kappa}$, where $T \subseteq \Upsilon^*$ is a tree with direction in $\Upsilon$, and $\kappa: T \rightarrow \Sigma$ labels each node by a letter in $\Sigma$. Thus, the definition is similar to the one introduced in \autoref{automata sol}, except that now, the labels are on the nodes. In order to distinguish between the two formalisms we use the term $\Sigma$-edge-labeled $\Upsilon$-tree for the labeled trees introduced in \autoref{automata sol}, and use $\ell$ and $\kappa$ for the edge- and node-labeling functions, respectively.  Thus, a $\Sigma$-edge-labeled $\Upsilon$-tree is a pair $\zug{T,\ell}$, where $T$ is as above and $\ell(x,x\cdot a)\in \Sigma$ for all $x\in T$ and $a\in \Upsilon$.

\stam{In the context reactive systems, directions typically correspond to the set of assignments to the input signals, and labels correspond to the assignments to both input and output signals. Thus, $\Upsilon=2^I$ and $\Sigma=2^{I\cup O}$. In that case, a $2^{I\cup O}$-node-labeled $2^I$-tree of the form $\zug{(2^I)^*,\kappa}$ is thought of as a \emph{computation tree} of a system that reads input signals in $I$ and responds with assignments to output signals in $O$. Note that for convenience, the labeling function $\kappa$ includes also information about the skeleton of the tree (the input assignments), and not only the assignments to the output signals. As such, the label of a node $i_1\cdots i_k\in (2^I)^*$ is $\kappa(i_1\cdots i_k)=i_k\cup o$, where $o$ is the assignment to the output signals in $O$ that the system generates when it reads the input sequence $i_1\cdots i_k$.}

%ofer3: environment generates-->initiates
For sets $I$ and $O$ of input and output signals, we consider $2^{I\cup O}$-labeled $2^I$-trees that are thought as a \emph{computation tree} of a system that reads assignments to  input signals in $I$ and responds with assignments to output signals in $O$. 
Note that, for convenience, the nodes are labeled by the assignments to both the input and output signals, thus $\kappa$ also includes information about the direction in which nodes are reached. In more details, the label of a node $i_1\cdots i_k\in (2^I)^*$ is $\kappa(i_1\cdots i_k)=i_k\cup o$, where $o$ is the assignment to the signals in $O$ that the system generates after it reads the input sequence $i_1\cdots i_k$. Traditionally, the root of the computation tree is labeled with the empty assignment $\kappa(\varepsilon)=\emptyset$. This is to reflect the fact that the environment initiates the interaction. 

Recall from \autoref{automata sol} that each TGE $\T$ induces a $(2^\C\times \P_{M,\H,\G})$-edge-labeled $2^\V$-tree $T_{\T}=\zug{(2^\V)^*,\ell}$, which encodes the interaction of $\T$ with all possible visible input sequences. 
%orna3
In \autoref{automata sol}, we only concerned the computations generated in $T_{\T}$. Here, we need to consider the {\em computation tree} of $\T$, whose branches correspond to assignments to both $V$ and $H$, and whose nodes are labeled by assignments to both $C$ and $G$. Thus, the computation tree of $\T$ is the $2^{I\cup O}$-labeled $2^{I}$-tree $\mathrm{full}(T_{\T})=\zug{(2^I)^*,\kappa}$, constructed by extending each branch of $T_{\T}$ with all possible assignments to the hidden input signals and by augmenting the labels to include assignments to all input and output signals. In particular, the assignment to the guided output signals in a node $w \in (2^I)^*$ are determined by the environment following the sequence of programs along $w$. 

In order to define $\kappa$ formally, we also define a memory labeling $\kappa_M:(2^I)^* \rightarrow M$ that maps each history $w \in (2^I)^*$ to the environment's memory state after processing $w$. We also extend the labeling function $\ell$ to the domain $(2^I)^*$, rather than $(2^V)^*$, by ignoring the hidden input signals: 
%ofer3: denoted x and x_V for shorter \ell...
%orna4: since it's used only two times, and it's ba tov ba'aiyn to see the v's only, I think it's better without it, and I also removed the x.
for all $v_1, \ldots, v_{k+1} \in 2^\V$ and $h_1, \ldots, h_{k+1} \in 2^\H$, 
%if $x=(v_1\cup h_1)\cdots (v_k\cup h_k)$ and $x_V=$, then we define 
we define
$\ell((v_1\cup h_1)\cdots (v_k\cup h_k), (v_1\cup h_1)\cdots (v_{k+1}\cup h_{k+1})) = \ell(v_1\cdots v_k, v_1\cdots v_{k+1})$.
%$\ell(x, x \cdot (v_{k+1}\cup h_{k+1})) = \ell(v_1\cdots v_k, v_1\cdots v_k \cdot v_{k+1})$.

The functions $\kappa$ and $\kappa_M$ are then defined inductively as follows. For the empty history $\varepsilon$, we define $\kappa(\varepsilon) = \emptyset$ and $\kappa_M(\varepsilon) = m_0$, where $m_0$ is the initial memory of the environment. For a non-empty history $w = w' \cdot (v \cup h) \in (2^I)^*$, where $w' \in (2^I)^*$, $v \in 2^\V$, and $h \in 2^\H$, let 
%orna3 type...
$\ell(w', w' \cdot (v \cup h)) = \zug{c, p}$, 
%$\ell(w', v \cup h) = \zug{c, p}$, 
$\kappa_M(w') = m'$, and $p(m', h) = \zug{m, g}$. Then, we define $\kappa(w) = v \cup h \cup c \cup g$ and $\kappa_M(w) = m$.

We evaluate the satisfaction of a \ctls formula $\spec$ in a TGE $\T$ by considering its  computation tree $\mathrm{full}(T_{\T})=\zug{(2^{I})^*,\kappa}$. Thus, $\T$ \emph{satisfies} $\spec$ iff $\mathrm{full}(T_{\T})$ satisfies $\spec$ according to the standard semantics of \ctls.

Our approach to SGE for branching-time specifications is similar to the automata-theoretic approach used in Section~\ref{automata sol} to solve SGE for \LTL specifications. In fact, one can view the solution to SGE for \LTL specifications as a special case of the solution to SGE for branching-time \ctls specifications. For \LTL specifications, the satisfaction of a formula in a computation tree requires all paths in the tree to satisfy the formula. Accordingly, the solution to SGE for \LTL specifications is based on the construction of a universal tree automaton. In contrast, for \ctls specifications, the relation among paths may be more complex, and involves both universal and existential path quantifiers. Thus, the solution to SGE for \ctls specifications is based on the use of the more general alternating tree automata.

In subsection~\ref{subsec:words-trees-automata}, we defined universal tree automata. Below we extend the definition to alternating tree automata. An \emph{alternating parity tree automaton} (APT) $\A=\zug{\Sigma,\Upsilon,Q,q_0,\eta,\alpha}$, with labels in $\Sigma$ and a set $\Upsilon$ of directions is similar to a universal tree automaton, but with a transition function $\eta:Q\times \Sigma\rightarrow B^+(Q\times \Upsilon)$, where $B^+(Q\times \Upsilon)$ is the set of positive Boolean formulas over $Q\times \Upsilon$. Note that unlike the definition in Subsection~\ref{subsec:words-trees-automata}, the transition function is with domain $Q\times \Sigma$ and not $Q\times \Sigma\times \Upsilon$. This reflects the fact we consider node-labeled, rather than edge-labeled, trees. This corresponds to the standard semantics of \ctls, where the labels are on states.

%orna3 some changes
A \emph{run} of $\A$ over a $\Sigma$-node-labeled $\Upsilon$-tree $\zug{T,\kappa}$ is a node-labeled tree $\zug{T_r,r}$ in which the 
%root is labeled $\zug{\varepsilon,q_0}$ and every other node is
each node is 
 labeled by an element of $T\times Q$. 
The root is labeled $\zug{\varepsilon,q_0}$.
Intuitively, every node $y\in T_r$ with label $r(y)=\zug{x,q}$ corresponds to a copy of $\A$ that reads the subtree of $T$ with root $x$ that visits state $q$. Accordingly, this copy of $\A$ branches to a set of successors that satisfy the formula specified in $\eta(q,\kappa(x))$. Note that since at most $\Upsilon\times Q$ copies are needed in order to satisfy the positive Boolean formula $\eta(q,\kappa(x))$, the branching degree of $T_r$ is at most $\Upsilon\times Q$. 
%ofer3: The "where $r: T_r \rightarrow T \times Q$ maps..." was already mentioned above and doesn't seem to be relevant. Kept only the first part about the directions.
%As such, $T_r$ can be thought as a tree with directions in $\Upsilon\times Q$, where $r: T_r \rightarrow T \times Q$ maps each node in $T_r$ to the corresponding node in $T$ and the state of $\A$ in that the corresponding copy of $\A$ visits when it reads this node. 
As such, $T_r$ can be thought as a tree with directions in $\Upsilon\times Q$.
Formally, $\zug{T_r,r}$ satisfies the following conditions:
\begin{itemize}
	\item $\varepsilon \in T_r$ and $r(\varepsilon)=\zug{\varepsilon,q_0}$.
	\item Consider a node $y \in T_r$ with $r(y)=\zug{x,q}$. Let $\eta(q,\kappa(x)) = \theta$. Then, there is a (possibly empty) set $S\subseteq \Upsilon\times Q$ that satisfies $\theta$, and for every $\zug{a,q'}\in S$, there is a child $z$ of $y$ such that $r(z)=\zug{x\cdot a,q'}$. In other words, the node $y$ branches into a set of directions in $\Upsilon\times Q$ that satisfies the transition function of $\A$ from $x$.
	%orna3: I stopped because perhaps you had a special reason not to go with the definition in earlier work. Also, I don't see how you proceed with a path not being a set of nodes. (that is, a few lines below, when \pi=	y_0 \cdot y_1 \cdots)
\end{itemize}

The acceptance condition of an APT is defined by a parity condition $\alpha: Q \rightarrow \mathbb{N}$, which assigns a priority to each state in $Q$. An infinite path $\pi$ in a run tree $T_r$ is \emph{accepting in $r$} if the highest priority that appears infinitely often among the $Q$ components of $r(\pi)$ is even. Formally, let $\pi=y_0,y_1,y_2,\ldots \in T_r$ be an infinite path in $T_r$, and let $r(y_i)=\zug{x_i,q_i}\in T\times Q$. Let $\inf(\pi|r)\subseteq Q$ be the set of states that $r$ visits infinitely often along $\pi$; that is, $q\in \inf(\pi|r)$ if and only if there are infinitely many indices $i$ such that $q_i=q$. The run $r$ on a path $\pi$ satisfies $\alpha$ if $\max(\alpha(q) :q\in \inf(\pi|r))$ is even. A run $\zug{T_r,r}$ of an APT $\A$ is \emph{accepting} if for all the paths $\pi$ in $T_r$, the run $r$ on $\pi$ satisfies $\alpha$. 
% We denote by $L(\A)$ the set of all $\Sigma$-labeled $\Upsilon$-trees that admit an accepting run of $\A$.

\begin{theorem}\cite{KVW00}
	\label{BV94-ctls-APT}
	Given a \ctls specification $\spec$ over $I\cup O$, we can construct an APT $\A_\spec$ of size $2^{O(|\spec|)}$ and index $3$ such that $\A_\spec$ accepts exactly all $2^{I\cup O}$-node-labeled $2^I$-trees that satisfy $\spec$.
\end{theorem}

We use the APT $\A_\spec$ in order to construct a TGE that realizes $\spec$. The idea is to construct an APT $\U_\spec$ that reads a $(2^{\C}\times \P_{M,\H,\G})$-edge-labeled $2^{\V}$-tree $T$, and let it simulate $\A_\spec$ on the induced computation tree $\mathrm{full}(T)$. Then, similar to the approach in Section~\ref{automata sol}, we reduce the synthesis of a TGE that realizes $\spec$ to the nonemptiness of $\U_\spec$. Essentially, an APT is nonempty iff its accepts a tree generated by a finite transducer, and such a tree,  accepted by $\U_\spec$, describes a TGE that satisfies $\spec$. Conversely, a TGE the satisfies $\spec$ defines a tree that is accepted by $\U_\spec$. The simulation is achieved by augmenting the states of $\A_\spec$ with the current memory and last input-output assignments. This way, although the APT $\U_\spec$ reads labels that do not include the assignments to the hidden and guided signals, the augmented states do include that missing information, which makes the simulation of $\A_\spec$ on the induced computation tree $\mathrm{full}(T)$ possible.
\begin{theorem}
\label{apt ctls}
	Given an APT $\A$ that reads $2^{I\cup O}$-node-labeled $2^I$-trees, we can construct an APT $\U$ that reads $(2^\C\times \P_{M,\H,\G})$-edge-labeled $2^\V$-trees, is of size $|\A|\cdot |M|\cdot 2^{|I|}$ and of the same index as $\A$, and accepts a tree $T=\zug{(2^\V)^*,\ell}$ iff $\mathrm{full}(T)$ is accepted by $\A$. 
\end{theorem}

\begin{proof}
Let $\A=\zug{2^{I\cup O},2^I,Q,q_0,\eta,\alpha}$ be an APT over $2^{I\cup O}$-node-labeled $2^I$-trees, we construct the APT $\U=\zug{(2^\C\times \P_{M,\H,\G}),2^{\V},Q',q'_0,\eta',\alpha'}$ as follows.
The states of $\U$ are defined as $Q' = Q \times M \times 2^{\V} \times 2^{\H}$. The initial state is defined as $q'_0=\zug{q_0,m_0,\emptyset,\emptyset}$, where $m_0$ is the initial memory of the environment. The acceptance condition $\alpha'$ is induced by $\alpha$ and is defined by $\alpha'(q,m,a)=\alpha(q)$ for all $q\in Q$, $m\in M$ and $a\in 2^{I\cup O}$. We continue by defining the transition function $\eta'$ so that $\U$ accepts a tree $T$ iff the induced computation tree $\mathrm{full}(T)$ is accepted by $\A$.

Note that while the inputs for $\U$ are edge-labeled trees, the input for $\A$ are node labeled. To overcome this, we think of a $(2^\C\times \P_{M,\H,\G})$-edge-labeled $2^\V$-tree $T=\zug{(2^{\V})^*,\ell}$ as node labeled, where labels are pushed downward to their target node. Formally, we extend $\ell$ to non-root nodes $w=w'\cdot v\in (2^\V)^*$ for $w'\in (2^\V)^*$ and $v\in 2^V$ by $\ell(w)=\ell(w',w'\cdot v)$. For the root node we define $\ell(\varepsilon)=\zug{\emptyset,p_0}$ where $p_0$ is the program defined $p_0(m,h)=\zug{m_0,\varepsilon}$ for all $m\in M$ and $h\in 2^\H$. This is to ensure that all copies of $\A_\spec$ at the first level of the tree start with the initial memory $m_0$. Note that at the root the assignments to the controlled and guided outputs are set both to $\emptyset$ which is aligned with with how the label of the root of the computation tree $\mathrm{full}(T)$ is defined. Indeed, $\kappa(\varepsilon)=\emptyset$.

The transition function $\eta':Q'\times (2^\C\times \P_{M,\H,\G})\rightarrow B^+(Q'\times 2^V)$ is defined for every $s=\zug{q,m,v,h}\in Q'$ and $\zug{c,p}\in 2^\C\times \P_{M,\H,\G}$ as follows. We think of the tuple $\zug{q,m,v,h}$ as a state that simulates the scenario of a copy of $\A$ reading a node of $\mathrm{full}(T)$ from state $q$ while the environment is at memory $m$ and its last input assignment was $v\cup h$. Recall that labels were pushed from edges downward to their target node. This means that from state $s=\zug{q,m,v,h}$ that is labeled with $\zug{c,p}$ we need to apply the program $p$ on $m$ and $h$ in order to extract the current assignment $g\in 2^\G$ to the guided outputs and the next memory state $m'\in M$ of the environment. Let $p(m,h)=\zug{m',g}$ and $a=v\cup h\cup c\cup g$. Thus, the corresponding copy of $\A$ is in state $q$ and reads a node labeled $a$. Therefore, we define $\eta'(\zug{q,m,v,h},\zug{c,p})$ to be the positive Boolean formula above $Q'\times 2^V$ that is induced by $\eta(q,a)$ by replacing every atomic proposition $\zug{q',v'\cup h'}\in Q\times 2^{\V\cup 
\H}$ that appears in $\eta(q,a)$ with the atomic proposition $\zug{\zug{q',m',v',h'},v'}$. 
 
We prove that a $(2^\C\times \P_{M,\H,\G})$-edge-labeled $2^\V$-tree $T$ is accepted by $\U_\spec$ iff the induced computation tree $\mathrm{full}(T)$ is accepted by $\A$.

To establish the correspondence between runs of $\U$ on $T$ and runs of $\A$ on $\mathrm{full}(T)$, observe the following invariant that is preserved along transitions edges of $\mathrm{full}(T)$: suppose that after processing the sequence $w=(v_1\cup h_1)\cdots(v_k\cup h_k)$ of visible and hidden inputs, the environment's memory is $\kappa_M(w) = m'$, and the label for the next edge in $T$ for the next visible input $v_{k+1}$ is $\ell(v_1 \cdots v_k, v_1\cdots v_{k+1}) = \zug{c, p}$. Then, for any next assignment $h_{k+1}$ to the hidden inputs, if applying the program $p$ yields $p(m', h_{k+1}) = \zug{ m, g}$, then the label at the corresponding next node in $\mathrm{full}(T)$ is $\kappa((v_1 \cup h_1) \cdots (v_k \cup h_k) \cdot (v_{k+1} \cup h_{k+1})) = v_{k+1} \cup h_{k+1} \cup c \cup g$, and the updated memory is $\kappa_M(w\cdot (v_{k+1}\cup h_{k+1})) = m$. Thus, a copy in a run tree of $\U$ that is in state $\zug{ q, m', v, h}$ and reads a node labeled $\zug{c,p}$ corresponds precisely to a copy in a run tree of $\A$ in state $q$ reading a node labeled $v \cup h \cup c \cup g$, where $m'$ is the current memory and $v\cup h$ is the current input assignment. This establishes the required one-to-one correspondence.

Finally, the acceptance condition of $\U$ is derived from the acceptance condition of $\A$ by projection. Hence, a run tree of $\U$ on $T$ is accepting iff the corresponding run tree of $\A$ on $\mathrm{full}(T)$ is accepting. This completes the proof.
\end{proof}

We can now conclude the section with the following complexity result for SGE with branching-time specifications.
\begin{theorem}
\label{sg-ctls}
	The \ctls SGE problem is 2EXPTIME complete. Given a \ctls formula $\spec$ over sets of signals $\V$, $H$, $C$, and $\G$, and a integer $k$, deciding whether $\spec$ is $(\V,H,C,\G)$-realizable by a TGE with memory $k$ can be done in time doubly exponential in $|\spec|$ and exponential in $k$. 	
%	Given a \ctls specification $\spec$ over $I\cup O$, the synthesis of a TGE that $\vhcg$-realizes $\spec$ can be reduced to the nonemptiness of an APT of size $2^{O(|\spec|)}\cdot |M|$ and index $3$. In particular, the synthesis problem is decidable in time 2EXPTIME in $|\spec|$ and EXPTIME in $|M|$.
\end{theorem}
\begin{proof}
The lower bound follows from the lower bound for LTL SGE (Theorem~\ref{bounded 2exp}) and the fact that \ctls is a superset of \LTL.\footnote{
Since \ctls satisfiability is 2EXPTIME-complete \cite{VS85}, one could also reduce from \ctls satisfiability. Indeed, a specification $\psi$ over a set $AP$ of atomic propositions is satisfiable iff $\psi$ is realizable by a TGE in which all the propositions in $AP$ are controllable output signals. The reduction, however, has to ensure that the computation tree of the TGE has a sufficient branching degree. For instance, if $I = \emptyset$, then the computation tree induced by a TGE has only one branch, which may not suffice for formulas with multiple existential path quantifiers.
By the ``sufficient branching-degree property'' of \ctls, a formula $\spec$ is satisfiable iff it is satisfiable in a tree with branching degree $\#_E(\spec) +1$, where $\#_E(\spec)$ is the number of existential quantifiers in $\spec$ \cite{ES84}. Accordingly, we reduce the satisfiability of $\spec$ to SGE of $\spec$ with 
$\left\lceil \log(\#_E(\spec) +1) \right\rceil$ arbitrary input signals, all in $V$, and $C=AP$.
}
We continue to the upper bound.
	By Theorem~\ref{BV94-ctls-APT}, there exists an APT $\A_\spec$ of size $|\A_{\spec}|=2^{O(|\spec|)}$ and index $3$ that accepts the set of $2^{I\cup O}$-labeled $2^I$-trees that satisfy $\spec$. Let $M$ be a set of size $k$. By Theorem~\ref{apt ctls}, we can use $\A_\spec$ to construct an APT $\U_\spec$ that reads $(2^\C\times \P_{M,\H,\G})$-edge-labeled $2^\V$-trees, is of size $|\A_\spec|\cdot |M|\cdot 2^{|I|}$ and of the same index as $\A_\spec$, and accepts a tree $T=\zug{(2^\V)^*,\ell}$ iff $\mathrm{full}(T)$ is accepted by $\A_\spec$. The size of $\U_\spec$ is thus $2^{O(|\spec|)}\cdot |M|$. It well known that the synthesis of a language specified by an APT can be reduced to its nonemptiness and that the nonemptiness problem of an APT of a fixed index is decidable in time exponential in it's size~\cite{KV97c,KV05c}. Thus, the synthesis of a TGE that $\vhcg$-realizes $\spec$ can be reduced to the nonemptiness of $\U_\spec$, which leads to the desired complexity result.
\end{proof}

\section{Discussion}

We introduced {\em synthesis with guided environments}, where the system can  utilize the knowledge and computational power of the environment. Straightforward directions for future research include extensions of the many settings in which synthesis has been studied to the ``the guided paradigm''. Here we discuss two directions that are more related to the paradigm itself.

\vspace{4mm}
{\bf Dynamic hiding and guidance.} In the setting studied here, the partition of $I$ and $O$ into visible, hidden, controlled, and guided signals is fixed throughout the computation. In some settings, these partitions may be dynamic. For example, when visibility depends on sensors that the system may activate and deactivate \cite{AKK19} or when signals are sometimes hidden in order to maintain the privacy of the system and the environment \cite{KL22}. The decision which signals to hide in each round may depend on the system (e.g., when it instructs the environment which signals to hide in order to maintain its privacy), the environment (e.g., when it prefers not to share sensitive information), or an external authority (e.g., when signals become hidden due to actual invisibility).
As for output signals, their guidance may depend on the history of the interaction (e.g., we may be able to assume amenability from the environment only after some password has been entered).

\stam{
Moreover, in the context of preserving privacy in computations, it is natural utilizing programs not only for the delegated signals and the memory update but also as a way for the system to send to the environment queries regarding its current input assignment. In our current setting the system always asks the environment to send its assignment to a constant set $\V$ of visible signals. Then, in dynamic hiding, this query encodes the set of signals that the system requires to disclose. But sometimes we can do even better. Consider for instance the case we need to realize $\spec=i_1\rightarrow(i_2\iff o)$ while preserving the value of $i_2$ hidden and controlling $o$. In the current setting the only thing we can do is request the environment to send the value of $i_2$, but this is insufficient for realizing $\spec$. Also a dynamic approach of hiding won't be useful. Yet, a system that requests from the environment to send the value of the predicate $i_1\iff i_2$, can assign that value to $o$ and satisfy $\spec$. Indeed, $i_1\rightarrow (i_2\iff (i_1\iff i_2))\equiv i_1\rightarrow i_1\equiv \True$. 
}

%Beyond coming up with the right definition for dynamic hiding and delegation, we expect the SGE algorithm to be more difficult.

%A framework for SGE with dynamic hiding and delegation would enable such a dynamic partition. Beyond coming up with the right definition for it, we expect the SGE algorithm to be more difficult.

\vspace{4mm}
{\bf Bounded SGE.}
SGE involves a memory that can be used by the environment. As in the study of traditional {\em bounded synthesis} \cite{SF07,KLVY11}, it is interesting to study SGE with given bounds not only on the memory for the guiding but also bounds on the state spaces of the system and the environment. In addition to better modeling the setting, the bounds are used in order to improve the complexity of the algorithm, and they can also serve in heuristics, as in SAT-based algorithms for bounded synthesis \cite{Ehl10}. In the setting of SGE, it is interesting to investigate the tradeoffs among the three involved bounds. It is easy to see that the two bounds that are related to the environment, namely the bound on its state space and the bound on the memory supervised by the system, are dual: an increase in the memory supervised by the system makes more specifications realizable, whereas an increase in the size of the state space of the environment makes fewer specifications realizable.

%ofer: Do you refer below to the size of the TGE or the size of the program? I think that it's the TGE, but it reads as it's the programs.
%orna1 changed
Another parameter that is interesting to bound is the number of different programs that a TGE may use, or the class of possible programs. In particular, restricting SGE to programs in which guided output signals can be assigned only the values of hidden signals or values stored in registers, will simplify an implementation of the algorithm. Likewise, the update of the memory during the interaction may be global and fixed throughout the computation.

\bibliographystyle{IEEEtran}
\bibliography{../../../ok}
\end{document}